\documentclass[aps,pra,showpacs,twoside,twocolumn,10pt]{revtex4-2}
\usepackage[colorlinks=true, citecolor=red, urlcolor=blue ]{hyperref}
\usepackage{times,epsfig,amssymb,amsfonts,amsmath,bm,subfigure,mathtools,amsthm,braket,soul,enumitem,color,physics,graphics,graphicx}
\usepackage[normalem]{ulem}
\usepackage{comment}
\usepackage{epstopdf}
\newtheorem{theorem}{Theorem}
\newtheorem{corollary}{Corollary}
\newtheorem{definition}{Definition}

\newtheorem{example}{Example}
\newtheorem{lemma}{Lemma}

\newtheorem{observation}{Observation}

\newcommand{\ph}[1]{{\color{blue}#1}}

\begin{document}

\title{
Entanglement-enhanced fluctuation-free daemonic ergotropy with random measurements
}

\author{Saptarshi Roy$^{1,2}$, Pritam Halder$^{3,4}$, Tamal Guha$^5$, Hyunseok Jeong$^2$}
\affiliation{\(^1\)Center for Quantum Engineering, Research, and Education, TCG CREST, Bidhan Nagar, Kolkata - 700091, India \\
\(^2\)NextQuantum Innovation Research Center, Department of Physics and Astronomy, Seoul National University, Seoul 08826, South Korea \\
\(^3\)Networked Quantum Devices Unit, Okinawa Institute of Science and Technology Graduate University, Okinawa, Japan \\
\(^4\)Harish-Chandra Research Institute, A CI of Homi Bhabha National Institute, Chhatnag Road, Jhunsi, Prayagraj - 211019, India
\\ \(^5\)Mathematical Institute, Slovak Academy of Sciences, \v{S}tef\'anikova 49,
814 73 Bratislava, Slovakia}

\begin{abstract}
Optimized daemonic ergotropy can make entanglement thermodynamically dispensable in measurement-assisted work extraction: for a fixed system marginal, quantum-classical states can reproduce the maximal work obtainable by optimizing the auxiliary measurement. We show that this equivalence is broken when the auxiliary measurement is randomized. For qudit-qubit quantum-classical states under Haar-random projective measurements, we derive upper bounds on the averaged daemonic gain and prove a gain-fluctuation trade-off, showing that any positive randomized gain necessarily entails measurement-induced fluctuations. In sharp contrast, a family of two-qubit entangled pure states attains the algebraic maximum of the gain allowed by a system marginal while remaining fluctuation-free for every measurement basis, yielding at least twice the gain achievable by any quantum-classical state with the same marginal. We further 
demonstrate that such conclusion holds when general two-qubit separable states are considered positioning randomized gain as a sufficient criterion for entanglement certification. Finally, we show that the entanglement advantage persists under partially randomized measurements sampled from a polar cap around the optimal basis. These results establish randomized daemonic ergotropy as a thermodynamic probe of entanglement and exhibit the connection between measurement-induced work fluctuations on the type of correlations: quantum entanglement vs classical.
\end{abstract}	

\maketitle

\section{Introduction}
\label{sec:intro}
{Investigation of the connection between quantum correlations and work extraction is one of the central themes of quantum thermodynamics \cite{Acin2015, Huber2015, Binder2018, Goold2016,Sapienza2019}. Considerable effort has been devoted to understanding whether quantum correlations can provide an advantage over classical correlations in work extraction from composite systems. In this context, it was shown that when ergotropy \cite{Allahverdyan2004} is taken as the mechanism of work extraction, the difference between globally extractable work and locally extractable work -- referred to as the ergotropic gap \cite{PerarnauLlobet2015,Huber2015, Alimuddin2019} -- bears signatures of entanglement. Such investigations have been extended to multipartite systems \cite{Huber2015, Mir2022} as well as to continuous-variable systems \cite{Friis2018, Mir2022, Mir2026, mrinmoy2026}. However, a limitation of such an investigation is the requirement to implement global operations.

On the other hand, ergotropic work can also be extracted by measuring one party of a bipartite system composed of the system $(S)$ and an auxiliary $(A)$. The ergotropy obtained from the post-measurement states of the system defines the framework of \emph{daemonic ergotropy} \cite{Francica2017}. This has emerged as an active area of research in recent years \cite{Francica2017,Manzano2018,Bernards2019}. A natural question in this context is: how do correlations influence the amount of extractable work? While early studies suggested that entanglement provides an advantage in work extraction assisted by measurements \cite{Francica2017}, it was later shown that such advantages can be entirely reproduced using classical correlations alone \cite{Manzano2018, Bernards2019, simonov2025activation}. 
Physically, this leads to a striking conclusion: for work extraction assisted by local measurements on one subsystem of a bipartite state, classical correlations already suffice to achieve the maximal daemonic ergotropy. Stronger forms of quantum correlations, such as entanglement, do not provide any additional advantage over classical correlations, in contrast to their typical role in numerous quantum information tasks.

\begin{figure}
    \centering
    \includegraphics[width=1.1\linewidth]{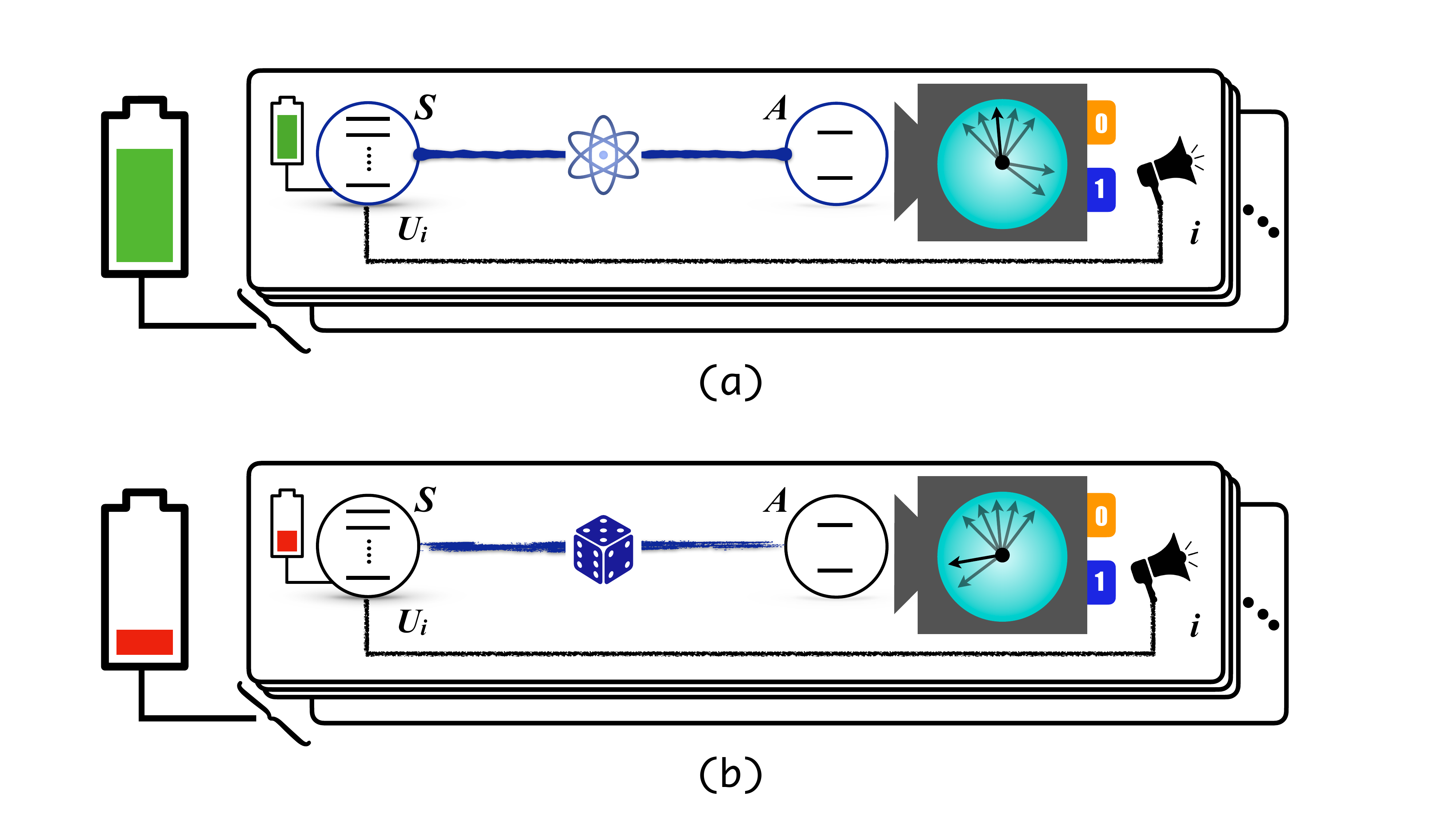}
    \caption{
    Schematic of the randomized daemonic ergotropy (RDE) extraction process. Random measurements on the auxiliary system $A$ assist ergotropy extraction from the system $S$, and the extracted ergotropy is averaged over the measurement ensemble. In contrast to the conventional setting with optimized assisting measurements, randomized measurements reveal a clear separation between the performance of $(a)$ entangled states and $(b)$ classically correlated states. Entanglement can provide both a larger average ergotropic gain and reduced measurement-induced fluctuations.}
    \label{fig:schematic}
\end{figure}

However, this equivalence between entanglement and classical correlations arises only when optimal measurements are performed on the auxiliary subsystem. Relying on such state-dependent optimal measurement scheme imposes significant operational constraint and represent a highly idealized scenario. In fact, the randomized measurements can arise naturally in realistic quantum-information protocols, where noise and imperfections affect the implementation of quantum operations and measurements.  Consequently, the performance of various genres of quantum information processing tasks including Bell inequality violations \cite{nv1, nv2, nv3}, quantum benchmarking \cite{Emerson2005,Knill2008}, state characterization \cite{Huang2020,Elben2019}, entanglement detection \cite{Tran_2015, Elben2020, Imai_2021} and quantification \cite{Eisfeld_2026}, quantum key distribution \cite{Laing_2010,Liu_2019} and estimation of quantum Fisher information \cite{Rath_2021} have recently been reassessed under randomized measurement. From an operational perspective, randomized measurements allow one to extract relevant properties of quantum states without requiring a shared reference frame or precise calibration. A comprehensive overview of the role of random measurements can be found in the recent review \cite{CieliskiReview_2024}. Crucially, in this framework, due to inherent randomness, examining statistical fluctuations becomes essential. While fluctuations in thermal devices such as quantum battery \cite{Imai_2023, Bakhshinezhad_2024, Sarkar_2025}, quantum heat engines and refrigerators \cite{Campisi2015, Segal_2018}, thermal transistors \cite{Das_2026} have been studied extensively studied via thermodynamic uncertainty relation \cite{Barato_2015,Horowitz_2020}, fluctuation-theorems \cite{Jarzynski_2004,Campisi_2011,Campisi_2014} etc., its impact on random measurement-assisted work extraction can evaluate robustness of the protocol.

In this work, we investigate the extraction of daemonic ergotropy when random projective measurements are performed on the auxiliary, a framework that we term as \textit{randomized daemonic ergotropy} (RDE), see Fig. \ref{fig:schematic}. We begin by analyzing the performance of quantum-classical (qc) states. This choice is motivated by the fact that qc states were shown to be optimal for daemonic ergotropy extraction under optimal (fine-tuned) measurements \cite{Francica2017, Bernards2019}. Thus, qc states provide a natural benchmark against which potential quantum advantages due to entanglement can be assessed when random measurements are performed. 

We first bound the maximal RDE gain achievable by qc states, and show that the fluctuations are lower-bounded by a quantity proportional to the daemonic gain.  This implies that qc states cannot support fluctuation-free daemonic gain. Remarkably, it leads to a fundamental no-go theorem, highlighting an intrinsic trade-off between gain and fluctuations in measurement-assisted work extraction for qc states.
Then we demonstrate that, unlike in the case of optimal measurements, in the case of randomized projective measurements, stronger quantum correlations like entanglement do indeed offer distinctive advantages over quantum-classical states, both in terms of the daemonic gain and fluctuations. In particular, we establish that entangled states can support both higher and fluctuation-free daemonic gains. 

Notably, we consider the case of general separable states and calibrate their performance. We demonstrate that no separable state can perform any better than the optimally chosen quantum classical states. This, in turn, enables us to provide a detector of entanglement in terms of the randomized demonic ergotropic gain. Additionally, we also show that there exist certain cases in which, although the randomized daemonic gain of any two states might be equal, an entangled state might support a lower value of fluctuations than the optimal separable states. This further highlights the importance of looking at fluctuations while investigating RDE. Finally, we show that in the scenario of partially randomized measurements around the optimal direction, the fluctuation-free gain by entangled states over qc states persists.

The article is organized in the following way. In Sec.~\ref{sec:rde}, we define the framework of RDE. We then proceed to evaluate the gain and fluctuation of RDE with qc states in Sec.~\ref{sec:optimalqc} as a benchmark. In Sec.~\ref{sec:rde_ent}, we demonstrate the role of entanglement in our framework alongwith a classification of states attaining fluctuation-free gains. In Sec.~\ref{sec:detector}, we provide a RDE-based sufficient criteria to detect two-qubit bipartite entangled states. In Sec.~\ref{sec:partial_rde}, similar analysis of daemonic ergotropy is performed with partially randomized measurement. Finally, we conclude in Sec.~\ref{sec:con}.}

\section{Randomized Daemonic Ergotropy}
\label{sec:rde}
{Consider an arbitrary quantum state $\rho^{SA}$ acting on a bipartite Hilbert space, $\mathbb C^{d_S}\otimes \mathbb C^{d_A}$ of local dimensions $d_k\forall k=S,A$ with $S$ and $A$ denoted as the system and the auxiliary, respectively, throughout this paper. Let $H$ be a reference Hamiltonian associated with $S$. The maximal amount of unassisted extractable work from $\rho^S=\Tr_A(\rho^{SA})$ under a cyclic unitary process, referred to as ergotropy \cite{Allahverdyan2004}, is given by ${\tt Erg}(\rho^S) = \text{Tr} (\rho^S H) - \min_{U} \text{Tr}(U\rho^S U^\dagger H)$ where the minimization is performed over all possible unitary transformations on $S$.} Now, the daemonic ergotropy obtained from the system $S$ of $\rho^{SA}$  by performing a projective measurement $\Pi = \{\Pi_i\}$ 
    on $A$ can be written as $E_D^{\rho^{SA}}(\Pi) := \sum_{i}{\tt Erg}(\rho^{SA}_{\Pi_i}),$ 
    where $\rho^{SA}_{\Pi_i} =$ Tr$_A\big( (\mathbb{I}\otimes \Pi_i) \rho^{SA} (\mathbb{I}\otimes \Pi_i) \big)$ is the un-normailized post-measurement state when $\Pi_i$ clicks.


\begin{definition}[Randomized Daemonic Ergotropy and gain]
The randomized daemonic ergotropy (RDE) $E_D$ \textcolor{black}{of the state \(\rho_{SA}\)} is the Haar averaged value of $E_D(\Pi)$, averaged over all possible projective measurements
\begin{eqnarray}
    E_D^{\rho^{SA}} := \int d \Pi ~\mu(\Pi) ~E^{\rho^{SA}}_D(\Pi),
    \label{eq:rde}
\end{eqnarray}
where $\mu(\Pi)$ denotes a probability density function over the space of projective measurements. 
The related quantity of interest, the randomized daemonic ergotropy gain, which is the amount of extra work that can be extracted from $S${, on average,} due to the random measurements performed at $A$, is given by
\begin{eqnarray}
    \Delta E_D^{\rho^{SA}} := E_D^{\rho^{SA}} - {\tt Erg}(\rho^S),
    \label{eq:dgain}
\end{eqnarray}
where $ {\tt Erg}(\rho^S)$ is the locally extractable ergotropy from the marginal state $\rho^S := \emph{Tr}_A (\rho^{SA})$ of the system $S$.
    \label{def:rde}
\end{definition}
\begin{lemma}
The randomized daemonic gain of any state $\rho^{SA}$ is bounded above by the difference of the energy and ergotropy of the marginal state: $\Delta E^{\rho^{SA}}_D \leq \Delta E_D^{\max}(\rho_S)$, where $\Delta E_D^{\max}(\rho_S) :=  \emph{Tr} (\rho^SH) - {\tt Erg}(\rho^S).$
    \label{lemma:absmax}
\end{lemma}
\begin{proof}
    For any measurement $\Pi$, we have $  
        E_D^{\rho^{SA}}(\Pi) \leq  \text{Tr} (\rho^SH).$
    Now, combining the above bound with Eq. \eqref{eq:rde}, the randomized daemonic ergotropy can be bounded.
    \begin{eqnarray}
        E_D^{\rho^{SA}} \leq \text{Tr} (\rho^SH) \int d \Pi ~\mu(\Pi) = \text{Tr} (\rho^SH).
    \end{eqnarray}
  The bound on the randomized gain thereafter follows from Eq. \eqref{eq:dgain}.  This completes the proof.
\end{proof}

\begin{definition}[Fluctuations in work extraction]
The fluctuation in RDE extraction is defined by the standard deviation of $\Delta E_D^{\rho^{SA}}(\Pi)$ as 
\begin{eqnarray}
    \mathcal{F}_D^{\rho^{SA}} = \sqrt{V_D^{\rho^{SA}}-\big(\Delta E_D^{\rho^{SA}}\big)^2}, 
    \label{eq:fluctuations}
\end{eqnarray}
where $V_D^{\rho^{SA}} := \int d\Pi~ \mu(\Pi) ~\big[\Delta E^{\rho^{SA}}_D(\Pi)\big]^2$.
\end{definition}

\section{Optimal vs. random measurements: \\ 
\hspace{0.75cm}the quantum-classical benchmark}
\label{sec:optimalqc}

Before presenting the results for randomized measurements, it is useful to recall what happens in the conventional daemonic protocol, in which one optimizes over all measurements \cite{Francica2017}. This establishes that the optimal states to be quantum-classical (qc).
A qc state has the form
\begin{equation}
    \xi^{SA}_{\mathrm{qc}}=\sum_j \sigma_j^S\otimes \ketbra{j}{j},
    \label{eq:qcform}
\end{equation}
where $\{\ket{j}_A\}$ is an orthonormal basis and the $\sigma_j^S$ are unnormalised positive operators. The key fact proved in Ref.~\cite{Bernards2019} is that for a given marginal state $\rho^S$ of the system, the maximal daemonic ergotropy is achievable by a qc-state as in Eq. \eqref{eq:qcform} with the following condition $\sum_j \sigma_j = \rho^S$. Moreover, for
such states the optimal measurement {on the auxiliary} is simply the projective measurement in the classical basis $\{\ket j \}$. 

In the randomized setting, we need to reassess the benchmark offered by qc states. 
For our analysis, {we} restrict the {auxiliary} dimension $d_A=2$, and  choose the probability distribution in Eq. \eqref{eq:rde}, $\mu(\Pi)$, to be the Haar measure, corresponding to the measurement directions sampled uniformly {over the Bloch sphere}.

\begin{theorem}[Upper bound of randomized daemonic gain for qc states]  
\label{th:maxergqc}
The randomized daemonic gain of any qudit-qubit quantum-classical state $\xi^{SA}$ with an arbitrary marginal state $\xi^S$ is upper bounded as 
    \begin{eqnarray}
     \nonumber\Delta  E^{\xi^{SA}}_D \leq \frac12 \Delta E_D^{\max}(\xi^S)= \frac{1}{2} \big[~\emph{Tr}(\xi^SH) - {\tt Erg}(\xi^S)\big],
    \end{eqnarray}
    where $H$ is the Hamiltonian of the system $S$.
\end{theorem}
\begin{proof}
 From Eq.~\eqref{eq:qcform}, a qudit-qubit qc state can be written as
\begin{eqnarray}
        \xi^{SA} = p\rho \otimes \ketbra{0}{0} + (1-p) \sigma &\otimes& \ketbra{0}{0}, \nonumber \\
      \text{where } ~ p\rho + (1-p)\sigma &=& \xi^S.
      \label{eq:defxi}
    \end{eqnarray}
In $d_A=2$, the projective measurements can be completely parameterized by two parameters, where any projective measurement
$\Pi$ corresponds to a measurement in some basis $\{\ket\psi,\ket{\psi^\perp}\}$ with $\ket{\psi} = \cos \frac{\theta}{2}\ket0+e^{i\phi}\sin \frac{\theta}{2}\ket1$ and 
$\ket{\psi^\perp} = \sin \frac{\theta}{2}\ket0 - e^{i\phi}\cos \frac{\theta}{2}\ket1.$
Here $0 \leq \theta \leq \frac\pi2$ and $0  \leq \phi \leq 2\pi$. Integrating $\theta$ up to $\pi/2$ is sufficient to capture all the orthogonal measurement pairs. 
Using the above parameterization, and the fact that we consider the measurements to be uniform, the RDE for any qudit-qubit state $\rho^{SA}$, from Def. \ref{def:rde} can be expressed as 
\begin{eqnarray}
    E_D^{\rho^{SA}} = \frac{1}{2\pi}\int_{0}^{2\pi} d\phi \int_0^{\frac\pi2} d\theta  ~\sin \theta ~E^{\rho^{SA}}_D(\theta,\phi).
    \label{eq:rdequbits}
\end{eqnarray}
When the auxiliary is measured in the $\{\ket\psi,\ket{\psi^\perp}\}$ basis, the un-normalized {post-measurement state of the system becomes } $\xi_{\psi (\psi^\perp)}^{SA} =$ $\Tr_A\big( (\mathbb{I}\otimes \ketbra{\psi (\psi^\perp)}{\psi (\psi^\perp)}) \xi^{SA} \big)$. Now $\xi_{\psi (\psi^\perp)}^{SA}$ depend solely on one of the measurement parameters $\theta$. By using the convexity of ergotropy \cite{Bernards2019}, we provide an upper bound of $ E_D^{\xi^{SA}}(\theta) = {\tt Erg}(\xi_\psi^{SA}) + {\tt Erg}(\xi_{\xi_{\psi^\perp}}^{SA})$ as
\begin{eqnarray}
    E^{\xi^{SA}}_D(\theta) \leq 2 \sin^2 \frac{\theta}{2} {\tt Erg}(\xi^S) &+& \cos \theta ~\text{Tr}(\xi^S H).
\end{eqnarray}
This upper bound constitutes the most important step of the proof. Now, from Eq. \eqref{eq:rdequbits}, we get $ E_D^{\xi^{SA}} \leq \frac12 \big[ \text{Tr} (\xi^SH) + {\tt Erg}(\xi^S) \big]$. Subtracting ${\tt Erg}(\xi^S)$ gives us the required bound on $\Delta E_D^{\xi^{SA}}$. A detailed proof is deferred to Appendix~\ref{app:prooft1}. 
\end{proof} 

\noindent When $\xi^S$ is a passive state, we get $\Delta E^{\xi^{SA}}_D = E^{\xi^{SA}}_D  \leq \frac12 \text{Tr} (\xi^SH).$ It follows directly from Eq. \eqref{eq:dgain} and from the fact that, for passive states, the ergotropy vanishes.



\subsection{NO fluctuation-free work extraction for qc states}
{\color{black}With an established bound for the gain, we now turn our attention to the fluctuation associated with the gain. In particular, }we investigate how much fluctuation the random measurements induce in the extracted work when the system and {auxiliary} possess qudit-qubit qc states. First, we present a no-go theorem which prohibits any fluctuation-free daemonic gain.

\begin{theorem}[Fluctuation lower bound for qc states]
The fluctuation in work extraction under randomized projective measurements of any qudit--qubit quantum-classical state $\xi^{SA}$ satisfies
the following bound
\begin{eqnarray}
   \mathcal F_D^{\xi^{SA}} 
   \geq
    \frac{1}{\sqrt{3}} \Delta E_D^{\xi^{SA}},
\end{eqnarray}
where $E_D^{\xi^{SA}}$ is the randomized daemonic gain of $\xi^{SA}$.
\label{th:fluctuationbound}
\end{theorem}

\noindent \textit{Sketch of the proof.}
 For a qudit--qubit qc state, any projective measurement on the auxiliary can be parametrized by a single parameter \(u=\cos\theta\in[0,1]\), which is uniformly distributed under the Haar measure. The corresponding daemonic gain can therefore be written as a function \(g(u):= \Delta E_D^{\xi^{SA}}(u) \). One observes that \(g(0)=0\), since an unbiased measurement yields identical post-measurement branches, and that \(g(u)\) is convex in \(u\), which is a direct consequence of ergotropy being convex on a mixture of states. Convexity together with the boundary condition \(g(0)=0\) implies that, for a fixed mean \(\int_0^1 g(u)\,du=\Delta E_D^{\xi^{SA}}\), the second moment satisfies \(\int_0^1 g(u)^2 du ~\geq\)  \(\frac{4}{3}(\Delta E_D^{\xi^{SA}})^2\), and hence \((\mathcal F_D^{\xi^{SA}})^2 \ge \frac{1}{3}(\Delta E_D^{\xi^{SA}})^2\).  A detailed proof is provided in Appendix~\ref{app:proofoft2}.

\noindent {\color{black} Physically, Theorem \ref{th:fluctuationbound} has a profound implication in the premise of RDE -- for any qc state with a finite gain, the minimum fluctuation it can support under ergotropic extraction grows with the gain it offers. Moreover, }this leads to the no-go condition that for any qc state $\xi^{SA}$ with \(\Delta E_D^{\xi^{SA}} > 0\), we have finite fluctuations \(\mathcal F_D^{\xi^{SA}} > 0\), thereby restricting any fluctuation-free randomized daemonic gain. Consequently, correlations of the qc state cannot support fluctuation-free universal daemonic work extraction: whenever random measurements yield a nonzero average gain, they necessarily induce a finite spread in the extractable work.

\begin{example}[State saturating the gain and fluctuation bounds]
Consider the qubit--qubit {qc} state
\begin{equation}
    \xi^{SA}_*
    =
    \frac12\,|0\rangle\!\langle0|_S\otimes |0\rangle\!\langle0|_A
    +
    \frac12\,|1\rangle\!\langle1|_S\otimes |1\rangle\!\langle1|_A,
\end{equation}
with maximally mixed system marginal $ \xi^S_*=\frac{\mathbb I}{2}$ and system Hamiltonian $ H = \ketbra{1}{1}.$
Note that for the considered Hamiltonian, \({\tt Erg}(\xi^S_*)=0\) and \(\Tr(\xi^S_* H)=1/2\). This state saturates both the bounds on gain and fluctuation in Theorems~\ref{th:maxergqc}  and \ref{th:fluctuationbound}.
\label{ex:1}
\end{example}

\noindent For projective measurement as described in the proof of Th. \ref{th:maxergqc}, the daemonic gain is given by $\Delta E_D^{\xi^{SA}_*}(\theta)=\frac{1}{2} \cos \theta.$
Averaging over Haar-random projective measurements, Eq. \eqref{eq:rdequbits}  yields $\Delta E_D^{\xi^{SA}_*}= \frac14 
    =
    \frac12\,\Tr(\xi^S_* H),$
showing that this state saturates the upper bound on the randomized daemonic gain for qc states with maximally mixed marginal. Moreover, 
\small
\begin{eqnarray}
    \Big(\mathcal F_D^{\xi^{SA}_*} \Big)^2
    =
 \frac14   \int_0^{\frac{\pi}{2}} \cos^2 \theta \sin \theta~ d\theta
    -
    \left(\Delta E_D^{\xi^{SA}_*}\right)^2 =
    \frac13\left(\Delta E_D^{\xi^{SA}_*}\right)^2. \nonumber
\end{eqnarray}
\normalsize
Therefore, it also saturates the fluctuation lower bound in the {Theorem}~\ref{th:fluctuationbound}. Hence, this state simultaneously realizes the maximal randomized daemonic gain and the minimal fluctuation compatible with that gain for qc states.


\section{Role of entanglement: \\
More gain with vanishing fluctuation}
\label{sec:rde_ent}

    

The natural question whether any correlation classes can overcome the bound provided by the qc states and how close they can come to the absolute ceiling. Interestingly, we observe a clear \textit{gap} of the extractable work when $S$ and $A$ are entangled to the case where they only share classical correlations.
We demonstrate this gap in the setting of two qubits and consider the system Hamiltonian $H = \ketbra{1}{1}$.
\begin{lemma}
    The randomized daemonic gain and fluctuations for the state $  \ket{\Psi^\alpha}=\sqrt{\alpha}\ket{01}-\sqrt{1-\alpha}\ket{10},$ where $0\leq \alpha\leq 1,$
is computed to be $\Delta E_D^{\Psi^\alpha} = \min \{\alpha,1 - \alpha\}, ~~\mathcal{F}_D^{\Psi^\alpha} = 0.$
    \label{lemma:psialpha}
\end{lemma}
\begin{proof}
    After measuring $\Psi^\alpha=\ket{\Psi^\alpha}\bra{\Psi^\alpha}$ in the basis $\Pi$, the un-normalized output states, $\Psi^\alpha_{\psi(\psi^\perp)}$  
     have ${\tt Erg}(\Psi_\psi^\alpha)=\cos^2\frac{\theta}{2}(1-\alpha)$ and ${\tt Erg}(\Psi_{\psi^\perp}^\alpha)=\sin^2\frac{\theta}{2}(1-\alpha)$ respectively. 
    Therefore, for any measurement basis $\Pi$, we have
    \begin{eqnarray}
      E^{\Psi^\alpha}_D(\Pi) =  {\tt Erg}(\Psi_\psi^\alpha)+ {\tt Erg}(\Psi_{\psi^\perp}^\alpha) = 1-\alpha. 
      \label{eq:ergpsiapi}
    \end{eqnarray}
    This right away implies that the randomized daemonic ergotropy of $\Psi^\alpha$ is given by $E^{\Psi^\alpha}_D  = 1-\alpha. $
Note that the {amount of} ergotropy that can be extracted from the marginal system state $\rho^S_\alpha = \text{Tr}_A \Psi^\alpha = \alpha\ketbra 0+(1-\alpha)\ketbra 1$ is given by 
    \begin{equation}
        {\tt Erg}(\rho^S_\alpha) = \Bigg\{
        \begin{array}{cc}
           0   & \alpha \geq 1/2, \\ 
            1-2\alpha  & \alpha < 1/2.
        \end{array}
        \label{eq:Epsialpha}
    \end{equation}
    Finally, by combining the above two results we have $\Delta E_D^{\Psi^\alpha} = \min\{\alpha,1-\alpha\}.$ The vanishing of $\mathcal{F}_D^{\Psi^\alpha}$ directly follows from Eq. \eqref{eq:ergpsiapi}.
    This completes the proof.
\end{proof}


We now contrast the above results with the ceilings for randomized {daemonic} gain of states with identical system marginals. We find that they saturate the absolute ceiling obtained in Lemma. \ref{lemma:absmax}.
\begin{corollary}
      For the given system marginal $\rho^S_\alpha = \alpha \ketbra{0}{0}+ (1-\alpha)\ketbra{1}{1}$, $\Psi^\alpha = \ketbra{\Psi^\alpha}{\Psi^\alpha}$ achieves the largest possible daemonic gain{, saturating the upper bound in Lemma~\ref{lemma:absmax}}.
\end{corollary}
\begin{proof}
    From Lemma~\ref{lemma:absmax} and for the system marginal $\rho^S_\alpha = \alpha \ketbra{0}{0}+ (1-\alpha)\ketbra{1}{1}$, we get $\text{Tr}(\rho^S_\alpha H) = 1-\alpha.$ Thereby, from Eq.~\eqref{eq:Epsialpha}, the upper bound reduces to $\min \{\alpha,1-\alpha\}$.
   This completes the proof.
\end{proof}
\noindent The above results establish that $\Psi^\alpha$ is the optimal state in terms of both the achievability of the maximal randomized daemonic gain for the given marginal state and the vanishing of fluctuations. {At this juncture, given the fixed marginal $\rho_\alpha^S$, the query is -- what advantage, if any, is conferred by the presence of entangling correlations  between $S$ and $A$ in the context of RDE? Below, we address this question quantitatively by demonstrating the superiority of quantum entanglement over classical correlations.} 
\begin{theorem}
    The randomized daemonic gain $\Delta E_D^{\Psi^\alpha}$ for an entangled pure state $\rho^{SA} =\Psi^\alpha$ with system marginal $\rho^S_\alpha$ is strictly greater than the gain offered by any qc-state with the identical system marginal, and satisfies $\Delta E_D^{\Psi^\alpha} \geq 2  \Delta E_D^{{\xi^{SA}}}.$
\end{theorem}
\begin{proof}
In Theorem~\ref{th:maxergqc}, we have shown that for qc-states $\xi^{SA}$ with system marginal $\xi^S = \rho^S_\alpha$, the daemonic gain is upper bounded by $  \Delta E_D^{\xi^{SA}} \leq \frac12 \big[ \text{Tr} (\rho^S_\alpha H) - {\tt Erg}(\rho^S_\alpha) \big].$
    For the considered system marginal, and using Lemma~\ref{lemma:psialpha}, we have 
    \begin{eqnarray}
        \Delta E_D^{\xi^{SA}} \leq \frac12 \min \{\alpha,1-\alpha\} = \frac{1}{2}\Delta E_D^{\Psi^\alpha}.
    \end{eqnarray}
\end{proof}

Conceptually, the above result implies that for work extraction under random projective measurements, entanglement offers more than double the randomized {daemonic} gain than any quantum classical state with identical system marginal. This demonstrates that, in contrast to the standard daemonic ergotropy scenario, where optimization over measurements renders entanglement operationally equivalent to classical correlations \cite{Bernards2019} -- \emph{randomization breaks this equivalence}. In our framework of work extraction under randomized measurements, entanglement offers a clear advantage over classical correlations. Additionally, as pointed out in Lemma \ref{lemma:psialpha}, the gain offered by $\Psi^\alpha$ is fluctuation-free.

\subsection{Classification of states with fluctuation-free gain}
An important question is whether there is a general classification of states which supports fluctuation-free randomized gain. Below, we provide a complete characterization.
\begin{theorem}[Two-qubit states with positive gain and zero fluctuations]
\label{th:classification}
Let \(\rho^{SA}\) be a two-qubit state written in Bloch form as $\rho^{SA}
    =
    \frac{1}{4}
    \Big[
        \mathbb I\otimes\mathbb I
        +
        \mathbf r\cdot\boldsymbol\sigma\otimes\mathbb I
        +
        \mathbb I\otimes \mathbf a\cdot\boldsymbol\sigma
        +
        \sum_{i,j=1}^3 T_{ij}\,\sigma_i\otimes\sigma_j
    \Big]. $
The necessary and sufficient condition for $\rho^{SA}$ to have \(\Delta E_D^{\rho^{SA}}>0\) and \(\mathcal F_D^{\rho^{SA}}=0\) is given by
\begin{eqnarray}
  TT^{\mathsf T} =\beta^2 I_3 + \mathbf r\,\mathbf r^{\mathsf T}, 
  \label{eq:conditionTmatrix}
\end{eqnarray}
for some constant $\beta>0$.
\label{th:fluc_free}
\end{theorem}

\noindent A detailed proof of the same is provided in the Appendix~\ref{app:proof t4}.
The above condition in Eq. \eqref{eq:conditionTmatrix} is satisfied for a wide variety of states, but we give the following example to highlight an important point.
\begin{example}[Depolarized pure states]
    Consider a two-parameter family of depolarized pure states
    \begin{equation}
    \rho^{SA}_{p,\alpha}
    =
    p|\psi^\alpha\rangle\!\langle\psi^\alpha|
    +
    \frac{1-p}{4}\mathbb I,
    \label{eq:dep-pure}
\end{equation}
where $ |\psi^\alpha\rangle =\sqrt{\alpha}|01\rangle - \sqrt{1-\alpha}|10\rangle.$ 
A quick calculation reveals the correlation matrix $T = \emph{diag}\{-2p\sqrt{\alpha(1-\alpha)},-2p\sqrt{\alpha(1-\alpha)},-p \}$, and $\mathbf r = \big(0,0, p(2\alpha-1)\big)^T$. These quantities satisfy Eq. \eqref{eq:conditionTmatrix} via $ TT^{\mathsf T} =4p^2 \alpha(1-\alpha) I_3 + \mathbf r\,\mathbf r^{\mathsf T}$, where we identify $\beta = 2p\sqrt{\alpha(1-\alpha)}$. Hence, due to Theorem~\ref{th:fluc_free}, these states should provide a finite fluctuation-free gain. This is also verified by evaluating the gain and fluctuation as
\begin{equation}
    \Delta E_D^{\rho^{SA}_{p,\alpha}}
    =
    p\min\{\alpha,1-\alpha\},
    \qquad\mathcal F_D^{\rho^{SA}_{p,\alpha}}=0.
\end{equation}
\end{example}
\noindent This  follows from Eq.~\eqref{eq:ergpsiapi} and by noting that ${\tt Erg}(q\sigma+(1-q)\mathbb{I}/2) = q ~{\tt Erg}(\sigma).$ 

Unlike the case of pure entangled states, the gains offered by the depolarized pure states in Eq.~\eqref{eq:dep-pure} are lower and in fact suppressed by a factor of the mixing probability $p$.
In particular, for $\alpha = 1/2$, we recover the well-known Werner states \cite{Werner_1989} for which the system marginal is maximally mixed, i.e.,
\begin{eqnarray}
 \rho_{p,\alpha=\frac12}^{SA} :=   \rho_{_W}(p) = p \ketbra{\psi^-}{\psi^-} + \frac{1-p}{4} \mathbb{I},
 \label{eq:wernerstate}
\end{eqnarray}
where $\ketbra{\psi^-}{\psi^-}:=\Psi^{\alpha=\frac12}$. The optimal quantum-classical states with a maximally mixed marginal (see Ex.~\ref{ex:1}), yield a randomized gain of $\Delta E_D^{\xi_*^{SA}} = 1/4$ and $\mathcal{F}_D^{\xi_*^{SA}} = \frac{1}{4\sqrt{3}}$. Therefore, for $p > \frac12$, the Werner states provide strictly better randomized gain compared to the optimal qc states, $\Delta E_D^{\rho_{_W}(p)} = \frac p2 > \frac{1}{4}$. In addition, for Werner states, the gain is fluctuation-free. Consequently, even for $p = \frac12,$ where the gains offered by the corresponding Werner state and the optimal qc state match, the Werner states can be established as being better owing to their vanishing fluctuations during work extraction. This again highlights the importance of considering fluctuations in our analysis, which helps us to filter out better states for work extraction from a set of states with identical randomized gain values.



\section{Certification of Entanglement using RDE}
\label{sec:detector}
In this section, we show how the randomized daemonic gain can be used as a detector of entanglement. To this end, we need to benchmark the gains offered by separable states of $2\times 2$ systems with a given system marginal.
Formally, to address this, we need to solve the following constrained optimization problem.
\begin{align}
\Delta E_D^{\mathrm{sep},\max}(\rho^S)
:=
\max_{\substack{
\rho^{SA} \in {\tt SEP} \\
\Tr_A[\rho^{SA}] = \rho^S
}}
\;\Delta E_D(\rho^{SA}) \,,
\label{eq:sepmax}
\end{align}
where $\Delta E_D^{\mathrm{sep},\max}(\rho^S)$ is the maximal gain offered by separable states with a fixed system marginal $\rho^S$. The above optimization problem becomes analytically tractable owing to the fact that both the objective function $\Delta E_D(\rho^{SA})$ and the feasible set of separable states with fixed marginal $\rho^S$ are convex. Therefore, the maximizer of $\Delta E_D(\rho^{SA})$ is an extreme point of the feasible set. This key observation enabled us to arrive at the following theorem.
\begin{theorem}
    The maximal daemonic gain obtained for any separable state with system marginal $\rho^S$ is given by 
    \begin{eqnarray}
        \Delta E_D^{\mathrm{sep},\max}(\rho^S) = \Delta E_D(\rho^{SA}_*),
    \end{eqnarray}
    where the form of $\rho^{SA}_*$ is given by
    \begin{eqnarray}
    \rho^{SA}_*     =
    \frac12\,|\psi_+\rangle\!\langle\psi_+|
    \otimes |0\rangle\!\langle0|
    +
    \frac12\,|\psi_-\rangle\!\langle\psi_-|
    \otimes |1\rangle\!\langle1|,
    \label{eq:optimalqcstate}
    \end{eqnarray}
    with $\frac12\,|\psi_+\rangle\!\langle\psi_+|
    +
    \frac12\,|\psi_-\rangle\!\langle\psi_-| = \rho^S.$ The corresponding value of the maximal gain $\Delta E_D^{\mathrm{sep},\max}(\rho^S)$ is given by
    \begin{eqnarray}
  \frac{1}{2}
\left[
\frac{1}{2}
+
\frac{s^2}{2\sqrt{1-s^2}}
\operatorname{arcsinh}
\left(
\frac{\sqrt{1-s^2}}{s}
\right)
-s
\right],
\label{eq:mrde}
    \end{eqnarray}
    where $\rho^S
=
\frac{1}{2}
\left(
\mathbb I_S+\mathbf{r}\cdot\boldsymbol{\sigma}
\right),
~s:=\|\mathbf{r}\|\in[0,1].$
\label{th:separable}
\end{theorem}
\noindent The proof is provided in Appendix \ref{app:methods}. The endpoints of $s=0,1$ are obtained by taking limits $0^+$ or $1^-$ of this expression.
The above theorem leads to the following observation:
\begin{observation}
Since the optimal separable state is quantum-classical, see Eq. \eqref{eq:optimalqcstate}, the bound obtained in Theorem~\ref{th:maxergqc} also holds for $2\times2$ separable states. Therefore, the maximal randomized daemonic gain of a separable state $\rho^{SA}$ with system marginal $\rho^S$ is bounded by $\Delta E_D^{\mathrm{sep},\max}(\rho^S) \leq \frac{1}{2} \Delta E_D^{\max}(\rho_S).$
    \label{ref:obs1}
\end{observation}
\noindent Next, we present a sufficient condition for certification of entanglement using RDE with the aid of Theorem~\ref{th:separable} and Obs.~\ref{ref:obs1}.
\begin{theorem}[Randomized daemonic gain as an entanglement detector]
For any $2\otimes 2$ state $\chi^{SA}$ with system marginal $\chi_S$, a sufficient condition for certifying entanglement: \\
If $\Delta E_D^{\chi^{SA}} > \Delta E_D^{\mathrm{sep},\max}(\chi^S) \implies$ $\chi^{SA}$ is entangled. 
\label{th:rde-detector}
\end{theorem}
To give a concrete visualization of our investigation, 
we characterize the maximal randomized daemonic gain achievable by separable states for a fixed passive system marginal, $ \rho^S
    =
    p\,|0\rangle\!\langle 0|
    +
    (1-p)\,|1\rangle\!\langle 1|,$ where $
     p\in\left[\frac12,1\right],$ see Fig. \ref{fig:sep-passive}. 
\begin{figure}[ht]
    \centering
    \includegraphics[width=\linewidth]{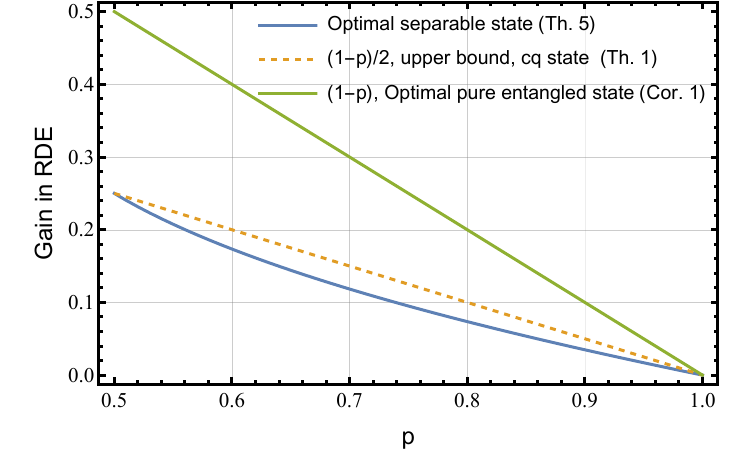}
    \caption{\textbf{Randomized daemonic gain with fixed system marginal $\rho^S=p\ketbra{0}{0}+(1-p)\ketbra{1}{1}$}. The maximal value of gain in RDE, obtained from general separable states via Theorem~\ref{th:separable} is plotted against $p$ (solid blue line). The linear bound for the cq states obtained in Theorem~\ref{th:maxergqc} is plotted via dashed orange lines. The green solid line is the maximal achievable gain for optimal pure entangled state (Corollary~\ref{lemma:absmax}). All axes are dimensionless.}
    \label{fig:sep-passive}
\end{figure}
Notably, Theorem \ref{th:rde-detector} provides a stronger criterion compared to the $\frac12 \Delta E_D^{\max}(\rho_S)$ bound reported in  Observation \ref{ref:obs1}. The bounds match for the case of maximal mixed marginal $\rho^S = \mathbb{I}/2$, $\rho^{SA}_{*}$ takes the form of the state in Example \ref{ex:1}. This immediately implies that, for the Werner states (see Eq. \eqref{eq:wernerstate}), entanglement can be certified for $p > \frac12$.

We end this section with an important observation. Although we have established the randomized demonic game to detect entanglement of Werner states for $p>\frac12$, does it imply that for $p=\frac12$, the corresponding Werner state does not possess any other identifiable signature owing to its non-vanishing entanglement? From our analysis in the previous sections, we know that the gains offered by Werner states are fluctuation-free. As a result, even for $p$ equal to half, for which the gains offered by the Werner state match the bound for the separable states, we can filter out the Werner states from their property of vanishing fluctuations. This again highlights the importance of investigating work extraction through the lens of fluctuations.

\section{Partially randomized measurements}
\label{sec:partial_rde}
Until this point of our manuscript, we have considered the measurements to be completely random, that is, they are drawn from a uniform distribution over the Bloch sphere. We now turn to a more structured scenario in which measurements are not fully random but are instead restricted to lie close to the optimal measurement configuration that yields the maximal daemonic gain. 
We model this feature by considering measurements drawn from a polar cap centered around the optimal measurement direction. A polar cap of angular (latitudinal) extent $\theta_0$ is described by the probability density function
\begin{equation}
p(\theta,\phi) =
\begin{cases}
\dfrac{1}{2\pi \left(1 - \cos \theta_0 \right)}, & 0 \le \theta \le \theta_0,\; 0 \le \phi < 2\pi, \\
0, & \text{otherwise},
\end{cases}
\label{eq:polar}
\end{equation}
where $\theta$ is the polar angle measured from $z$-axis, since we consider the optimal measurements to be the $\{\ket0, \ket 1\}$ basis, and $\phi$ is the azimuthal angle.
This distribution is uniform over the spherical cap and normalized with respect to the standard measure on the Bloch sphere, i.e., $\int_0^{2\pi} d\phi \int_0^\pi d\theta \, \sin\theta \, p(\theta,\phi) = 1.$ Such distributions have been investigated in other spheres of quantum information, especially in quantum teleportation \cite{Roy2022}.

\subsection{The quantum classical benchmark}
The central issue of investigation in this section is to establish a benchmark for the daemonic gain and the corresponding fluctuations achievable by quantum-classical states when measurements are drawn from the above polar-cap distribution. In particular, this benchmark quantifies the maximal advantage attainable in the absence of quantum correlations beyond classical correlations, under restricted measurement resources. This will allow us to meaningfully compare the performance of general quantum states against this classical baseline when measurements are confined to a neighbourhood of the optimal configuration.

\begin{theorem}[Bounds on gain and fluctuation for qc states]
    The randamized daemonic gain of a qudit-qubit quantum classical state $\xi^{SA}$ with marginal state $\xi^S$ when the measurements are chosen from a polar cap of latitudinal extent of $\theta_0$ is bounded by
    \begin{eqnarray}
        \Delta E_D^{\xi^{SA}}[\theta_0] \leq \cos^2 \frac{\theta_0}{2} \big[~\emph{Tr}(\xi^SH) - {\tt Erg}(\xi^S)\big].
        \label{eq:gaintheta}
    \end{eqnarray}
    The corresponding fluctuation in work extraction is lower bounded as follows
    \begin{eqnarray}
   \mathcal F_D^{\xi^{SA}}[\theta_0]
   \geq
    \frac{\tan^2 \frac{\theta_0}{2}}{\sqrt{3}} \Delta E_D^{\xi^{SA}}[\theta_0].
\end{eqnarray}    
\label{th:boundpolarcap}
\end{theorem}
\noindent The details of the proof are provided in the Appendix \ref{app:prooft6}. The flow of logic bears a lot of resemblance to the proofs in Theorems \ref{th:maxergqc} and \ref{th:fluctuationbound}. Importantly, Theorem \ref{th:boundpolarcap} recovers the results for the case when completely randomized measurements are employed for work extraction when we substitute $\theta_0 = \pi/2$.

The above bounds are satisfied for maximally mixed marignal by the same quantum-classical state as in Example \ref{ex:1}. The corresponding values of randomized daemonic gain and fluctuation are given by $E_D^{\xi^{SA}_*}[\theta_0]=\frac12 \cos^2\frac{\theta_0}{2}$ and $\mathcal F_D^{\xi^{SA}_*}[\theta_0]= \frac{1}{2\sqrt{3}}\sin^2 \frac{\theta_0}{2}= \frac{\tan^2 \frac{\theta_0}{2}}{\sqrt{3}}E_D^{\xi^{SA}_*}[\theta_0].$

\subsection{Advantages of entanglement}
In this section, we recalibrate the advantages offered by entanglement over qc states when the measurements are chosen from a polar cap, as in Eq. \eqref{eq:polar}. Like in the previous section, the advantages of entanglement are benchmarked against non-maximally entangled pure states state $  \ket{\Psi^\alpha}=\sqrt{\alpha}\ket{01}-\sqrt{1-\alpha}\ket{10},$ where $0\leq \alpha\leq 1.$

\begin{observation}[Enhancement in RDE]
\label{obs:enhancementpolar}
When the measurement employed for work extraction is restricted to the polar cap of latitudinal extent $\theta_0$, the randomized daemonic gain $\Delta E_D^{\Psi^\alpha}$ for an entangled pure state $\rho^{SA} =\Psi^\alpha$ with system marginal $\rho^S_\alpha$ is strictly greater than the gain offered by any qc-state with the identical system marginal, and satisfies the following inequality
    \begin{eqnarray}
        \Delta E_D^{\Psi^\alpha}{[\theta_0]} \geq \frac{2}{1+\cos \theta_0}  ~\Delta E_D^{{\xi^{SA}}}[\theta_0].
    \end{eqnarray}
\end{observation}
\noindent This follows directly from Lemma \ref{lemma:psialpha} and Eq. \eqref{eq:gaintheta} of Theorem \ref{th:boundpolarcap}. Interestingly, the advantages of entanglement persist for any finite extent of the polar cap, $\theta_0>0.$ Incidentally, the maximal gap is achieved for $\theta_0 = \frac{\pi}{2},$ i.e., the uniform case. Moreover, since Theorem~\ref{th:classification} directly holds independent of the probability density function, the general classification of fluctuation-free finite gain states holds for any distribution of the measurements. What changes is the difference between the gains offered by states possessing quantum correlations compared to the quantum-classical benchmark. For pure entangled states, the advantage in gain for the polar cap distribution is reported in Observation \ref{obs:enhancementpolar}.




\section{Conclusion}
\label{sec:con}
The role of quantum correlations between the system and the auxiliary in work extraction demands critical assessment, not only to unlock potential technological applications but also because of their fundamental importance in understanding quantum thermodynamics. Measurement on the correlated auxiliary can induce a higher value of ergotropy, known as daemonic ergotropy \cite{Francica2017}. However, in this realm of daemonic work extraction, performing state-dependent optimal measurement scheme banishes the role of bipartite entanglement \cite{Bernards2019}, as separable quantum-classical (qc) states are sufficient to yield maximum achievable ergotropy. In contrast, this work has analyzed the daemonic ergotropic gain via randomized projective measurements, a framework we term as \textit{randomized daemonic ergotropy (RDE)}. 

In this setting, we benchmark the gains and corresponding measurement-induced fluctuations by the qc states via analytical bounds. In particular, we establish a no-go theorem for qc states forbidding them to have any ``fluctuation-free gains".  
Interestingly, our investigation asserts the superiority of entanglement over classical correlations in this randomized regime, both in terms of gain and fluctuations. In particular, the best gains provided by entangled states are guaranteed to be at least twice those of any qc states with identical system marginal. At the same time, the gain obtained is also demonstrated to be fluctuation-free.
For $2 \times 2$ systems, we also rigorously show that the superiority of entanglement continues over the entire class of separable states and the same bound obtained for the qc states still holds. 
Consequently, our results provide an operational characterization of RDE via a sufficient criterion to detect bipartite entangled states.
Furthermore, we also provide a complete characterization $2 \times 2$ states that support fluctuation-free gains. For instance, there exist parameter regimes of the Werner state that yield a greater (or equal) work gain while simultaneously maintaining fluctuation-free work than the optimal qc state with identical marginal.

Finally, from a realistic perspective of the inability to perform fine-tuned measurements, we have analyzed the daemonic ergotropic gain with \textit{partially} randomized measurements which are performed uniformly around the optimal direction. Remarkably, both the no-go theorem prohibiting fluctuation-free work extraction from qc states and the advantages conferred by entanglement in enhancing ergotropic gain or enabling zero-fluctuation work extraction remain valid under arbitrary measurement randomizations. What changes is the scaling factors that, in turn, depend on the extent of randomization.

Our work thus paves the way to explore and reassess the role of various quantum correlations, specifically entanglement, in the premise of daemonic work extraction. It provides a sufficient condition for detection of entanglement and highlights the qualitative change of physics when one relaxes the rather artificial configuration of optimal measurements and works in the framework of RDE. It would be an interesting area of future study to address the impact of higher-dimensional auxiliary systems and more generalized random measurements.

\section*{Acknowledgements}
The authors thank Paranjoy Chaki for insightful discussions. This work was supported by the National Research Foundation of Korea (NRF) grants funded by the Korea government (MSIT) (No.~RS-2024-00413957 and No. RS-2024-00438415), and the Institute of Information \& Communications Technology Planning \& Evaluation (IITP) grant funded by the Korea government (MSIT) (IITP-2025-RS-2020-II201606 and IITP-2025-RS-2024-00437191). TG is supported by the Slovak Research and Development Agency through Grant no. APVV-22-0570, the Scientific Grant Agency of the Ministry of Education, Slovak Republic through Grant no. VEGA 2/0128/24 and the Štefan Schwarz Support Fund 2025/OV1/046 by the Slovak Academy of Sciences. Parts of a couple of proofs were verified using GPT-5.6 Sol.

\bibliography{ref.bib}

@book{Binder2018,
  editor    = {Binder, Felix and Correa, Luis A. and Gogolin, Christian and Anders, Janet and Adesso, Gerardo},
  title     = {Thermodynamics in the Quantum Regime: Fundamental Aspects and New Directions},
  series    = {Fundamental Theories of Physics},
  volume    = {195},
  publisher = {Springer},
  address   = {Cham},
  year      = {2018},
  doi       = {10.1007/978-3-319-99046-0},
  url       = {https://doi.org/10.1007/978-3-319-99046-0}
}

@article{Goold2016,
  author  = {Goold, John and Huber, Marcus and Riera, Arnau and del Rio, L{\'i}dia and Skrzypczyk, Paul},
  title   = {The role of quantum information in thermodynamics---a topical review},
  journal = {Journal of Physics A: Mathematical and Theoretical},
  volume  = {49},
  number  = {14},
  pages   = {143001},
  year    = {2016},
  doi     = {10.1088/1751-8113/49/14/143001},
  url     = {https://doi.org/10.1088/1751-8113/49/14/143001}
}

@article{Allahverdyan2004,
  author  = {Allahverdyan, Armen E. and Balian, Roger and Nieuwenhuizen, Theo M.},
  title   = {Maximal work extraction from finite quantum systems},
  journal = {Europhysics Letters},
  volume  = {67},
  number  = {4},
  pages   = {565--571},
  year    = {2004},
  doi     = {10.1209/epl/i2004-10101-2},
  url     = {https://doi.org/10.1209/epl/i2004-10101-2}
}

@article{PerarnauLlobet2015,
  author  = {Perarnau-Llobet, Mart{\'i} and Hovhannisyan, Karen V. and Huber, Marcus and Skrzypczyk, Paul and Brunner, Nicolas and Ac{\'i}n, Antonio},
  title   = {Extractable Work from Correlations},
  journal = {Physical Review X},
  volume  = {5},
  number  = {4},
  pages   = {041011},
  year    = {2015},
  doi     = {10.1103/PhysRevX.5.041011},
  url     = {https://doi.org/10.1103/PhysRevX.5.041011}
}

@article{Huber2015,
  author  = {Huber, Marcus and Perarnau-Llobet, Mart{\'i} and Hovhannisyan, Karen V. and Skrzypczyk, Paul and Kl{\"o}ckl, Claude and Brunner, Nicolas and Ac{\'i}n, Antonio},
  title   = {Thermodynamic cost of creating correlations},
  journal = {New Journal of Physics},
  volume  = {17},
  pages   = {065008},
  year    = {2015},
  doi     = {10.1088/1367-2630/17/6/065008},
  url     = {https://doi.org/10.1088/1367-2630/17/6/065008}
}

@article{Friis2018,
  author  = {Friis, Nicolai and Huber, Marcus},
  title   = {Precision and Work Fluctuations in Gaussian Battery Charging},
  journal = {Quantum},
  volume  = {2},
  pages   = {61},
  year    = {2018},
  doi     = {10.22331/q-2018-04-23-61},
  url     = {https://doi.org/10.22331/q-2018-04-23-61}
}

@article{Francica2017,
  author  = {Francica, Gianluca and Goold, John and Plastina, Francesco and Paternostro, Mauro},
  title   = {Daemonic ergotropy: enhanced work extraction from quantum correlations},
  journal = {npj Quantum Information},
  volume  = {3},
  pages   = {12},
  year    = {2017},
  doi     = {10.1038/s41534-017-0012-8},
  url     = {https://doi.org/10.1038/s41534-017-0012-8}
}

@article{Manzano2018,
  author  = {Manzano, Gonzalo and Plastina, Francesco and Zambrini, Roberta},
  title   = {Optimal Work Extraction and Thermodynamics of Quantum Measurements and Correlations},
  journal = {Physical Review Letters},
  volume  = {121},
  number  = {12},
  pages   = {120602},
  year    = {2018},
  doi     = {10.1103/PhysRevLett.121.120602},
  url     = {https://doi.org/10.1103/PhysRevLett.121.120602}
}

@article{Elben2020,
  author  = {Elben, Andreas and Vermersch, Beno{\^i}t and Brydges, Tim and Gong, Zhe-Xuan and Linke, Norbert M. and Jurcevic, Petar and Lanyon, Ben P. and Zoller, Peter and Blatt, Rainer and Roos, Christian F.},
  title   = {Mixed-State Entanglement from Local Randomized Measurements},
  journal = {Physical Review Letters},
  volume  = {125},
  number  = {20},
  pages   = {200501},
  year    = {2020},
  doi     = {10.1103/PhysRevLett.125.200501},
  url     = {https://doi.org/10.1103/PhysRevLett.125.200501}
}

@article{Emerson2005,
  author  = {Emerson, Joseph and Alicki, Robert and {\.Z}yczkowski, Karol},
  title   = {Scalable Noise Estimation with Random Unitary Operators},
  journal = {Journal of Optics B: Quantum and Semiclassical Optics},
  volume  = {7},
  number  = {10},
  pages   = {S347--S352},
  year    = {2005},
  doi     = {10.1088/1464-4266/7/10/021},
  url     = {https://doi.org/10.1088/1464-4266/7/10/021}
}

@article{Knill2008,
  author  = {Knill, E. and Leibfried, D. and Reichle, R. and Britton, J. and Blakestad, R. B. and Jost, J. D. and Langer, C. and Ozeri, R. and Seidelin, S. and Wineland, D. J.},
  title   = {Randomized Benchmarking of Quantum Gates},
  journal = {Physical Review A},
  volume  = {77},
  number  = {1},
  pages   = {012307},
  year    = {2008},
  doi     = {10.1103/PhysRevA.77.012307},
  url     = {https://doi.org/10.1103/PhysRevA.77.012307}
}

@article{Huang2020,
  author  = {Huang, Hsin-Yuan and Kueng, Richard and Preskill, John},
  title   = {Predicting many properties of a quantum system from very few measurements},
  journal = {Nature Physics},
  volume  = {16},
  pages   = {1050--1057},
  year    = {2020},
  doi     = {10.1038/s41567-020-0932-7},
  url     = {https://doi.org/10.1038/s41567-020-0932-7}
}

@article{Elben2019,
  author  = {Elben, Andreas and Vermersch, Beno{\^i}t and Dalmonte, Marcello and Cirac, J. Ignacio and Zoller, Peter},
  title   = {Statistical correlations between locally randomized measurements: A toolbox for probing entanglement in many-body quantum states},
  journal = {Physical Review A},
  volume  = {99},
  number  = {5},
  pages   = {052323},
  year    = {2019},
  doi     = {10.1103/PhysRevA.99.052323},
  url     = {https://doi.org/10.1103/PhysRevA.99.052323}
}

@article{CieliskiReview_2024,
   title={Analysing quantum systems with randomised measurements},
   volume={1095},
   ISSN={0370-1573},
   url={http://dx.doi.org/10.1016/j.physrep.2024.09.009},
   DOI={10.1016/j.physrep.2024.09.009},
   journal={Physics Reports},
   publisher={Elsevier BV},
   author={Cieśliński, Paweł and Imai, Satoya and Dziewior, Jan and Gühne, Otfried and Knips, Lukas and Laskowski, Wiesław and Meinecke, Jasmin and Paterek, Tomasz and Vértesi, Tamás},
   year={2024},
   month=Dec, pages={1–48} }

@article{Sapienza2019,
   author = {Facundo Sapienza and Federico Cerisola and Augusto J. Roncaglia},
   doi = {10.1038/s41467-019-10572-8},
   issn = {2041-1723},
   issue = {1},
   journal = {Nature Communications},
   month = {6},
   pages = {2492},
   title = {Correlations as a resource in quantum thermodynamics},
   volume = {10},
   year = {2019}
}

@article{Acin2015,
  title = {Extractable Work from Correlations},
  author = {Perarnau-Llobet, Mart\'{\i} and Hovhannisyan, Karen V. and Huber, Marcus and Skrzypczyk, Paul and Brunner, Nicolas and Ac\'{\i}n, Antonio},
  journal = {Phys. Rev. X},
  volume = {5},
  issue = {4},
  pages = {041011},
  numpages = {14},
  year = {2015},
  month = {Oct},
  publisher = {American Physical Society},
  doi = {10.1103/PhysRevX.5.041011},
  url = {https://link.aps.org/doi/10.1103/PhysRevX.5.041011}
}

@article{Rath_2021,
  title = {Quantum Fisher Information from Randomized Measurements},
  author = {Rath, Aniket and Branciard, Cyril and Minguzzi, Anna and Vermersch, Beno\^{\i}t},
  journal = {Phys. Rev. Lett.},
  volume = {127},
  issue = {26},
  pages = {260501},
  numpages = {6},
  year = {2021},
  month = {Dec},
  publisher = {American Physical Society},
  doi = {10.1103/PhysRevLett.127.260501},
  url = {https://link.aps.org/doi/10.1103/PhysRevLett.127.260501}
}

@article{Bernards2019,
  title = {Daemonic Ergotropy: Generalised Measurements and Multipartite Settings},
  volume = {21},
  ISSN = {1099-4300},
  url = {http://dx.doi.org/10.3390/e21080771},
  DOI = {10.3390/e21080771},
  number = {8},
  journal = {Entropy},
  publisher = {MDPI AG},
  author = {Bernards,  Fabian and Kleinmann,  Matthias and G\"{u}hne,  Otfried and Paternostro,  Mauro},
  year = {2019},
  month = Aug,
  pages = {771}
}

@article{Mir2022,
  title = {Thermodynamic Signatures of Genuinely Multipartite Entanglement},
  author = {Puliyil, Samgeeth and Banik, Manik and Alimuddin, Mir},
  journal = {Phys. Rev. Lett.},
  volume = {129},
  issue = {7},
  pages = {070601},
  numpages = {7},
  year = {2022},
  month = {Aug},
  publisher = {American Physical Society},
  doi = {10.1103/PhysRevLett.129.070601},
  url = {https://link.aps.org/doi/10.1103/PhysRevLett.129.070601}
}

@article{Mir2026,
  title = {Ergotropic Characterization of Continuous-Variable Entanglement},
  author = {Polo-Rodr\'{\i}guez, Beatriz and Centrone, Federico and Adesso, Gerardo and Alimuddin, Mir},
  journal = {Phys. Rev. Lett.},
  volume = {136},
  issue = {5},
  pages = {050201},
  numpages = {7},
  year = {2026},
  month = {Feb},
  publisher = {American Physical Society},
  doi = {10.1103/43jm-qkz2},
  url = {https://link.aps.org/doi/10.1103/43jm-qkz2}
}

@misc{mrinmoy2026,
      title={Hierarchies of Gaussian multimode entanglement from thermodynamic quantifiers}, 
      author={Mrinmoy Samanta and Sudipta Mondal and Ayan Patra and Saptarshi Roy and Aditi Sen De},
      year={2026},
      eprint={2602.18816},
      archivePrefix={arXiv},
      primaryClass={quant-ph},
      url={https://arxiv.org/abs/2602.18816}, 
}

@article{Alimuddin2019,
  title = {Bound on ergotropic gap for bipartite separable states},
  author = {Alimuddin, Mir and Guha, Tamal and Parashar, Preeti},
  journal = {Phys. Rev. A},
  volume = {99},
  issue = {5},
  pages = {052320},
  numpages = {17},
  year = {2019},
  month = {May},
  publisher = {American Physical Society},
  doi = {10.1103/PhysRevA.99.052320},
  url = {https://link.aps.org/doi/10.1103/PhysRevA.99.052320}
}

@article{nv1,
  title = {Measure of nonlocality which is maximal for maximally entangled qutrits},
  author = {Fonseca, E. A. and Parisio, Fernando},
  journal = {Phys. Rev. A},
  volume = {92},
  issue = {3},
  pages = {030101},
  numpages = {5},
  year = {2015},
  month = {Sep},
  publisher = {American Physical Society},
  doi = {10.1103/PhysRevA.92.030101},
  url = {https://link.aps.org/doi/10.1103/PhysRevA.92.030101}
}

@article{nv2,
  title = {Volume of violation of Bell-type inequalities as a measure of nonlocality},
  author = {Barasi\ifmmode \acute{n}\else \'{n}\fi{}ski, Artur and Nowotarski, Mateusz},
  journal = {Phys. Rev. A},
  volume = {98},
  issue = {2},
  pages = {022132},
  numpages = {9},
  year = {2018},
  month = {Aug},
  publisher = {American Physical Society},
  doi = {10.1103/PhysRevA.98.022132},
  url = {https://link.aps.org/doi/10.1103/PhysRevA.98.022132}
}

@article{nv3,
  title = {Towards an equivalence between maximal entanglement and maximal quantum nonlocality},
  volume = {20},
  ISSN = {1367-2630},
  url = {http://dx.doi.org/10.1088/1367-2630/aaca22},
  DOI = {10.1088/1367-2630/aaca22},
  number = {6},
  journal = {New Journal of Physics},
  publisher = {IOP Publishing},
  author = {Lipinska,  Victoria and Curchod,  Florian J. and Máttar,  Alejandro and Acín,  Antonio},
  year = {2018},
  month = June,
  pages = {063043}
}

@article{Roy2022,
  title = {Gain in performance of teleportation with uniformity-breaking distributions},
  author = {Roy, Saptarshi and Mal, Shiladitya and Sen(De), Aditi},
  journal = {Phys. Rev. A},
  volume = {105},
  issue = {2},
  pages = {022610},
  numpages = {12},
  year = {2022},
  month = {Feb},
  publisher = {American Physical Society},
  doi = {10.1103/PhysRevA.105.022610},
  url = {https://link.aps.org/doi/10.1103/PhysRevA.105.022610}
}

@article{Laing_2010,
  title = {Reference-frame-independent quantum key distribution},
  author = {Laing, Anthony and Scarani, Valerio and Rarity, John G. and O'Brien, Jeremy L.},
  journal = {Phys. Rev. A},
  volume = {82},
  issue = {1},
  pages = {012304},
  numpages = {5},
  year = {2010},
  month = {Jul},
  publisher = {American Physical Society},
  doi = {10.1103/PhysRevA.82.012304},
  url = {https://link.aps.org/doi/10.1103/PhysRevA.82.012304}
}

@article{Liu_2019,
  title = {Reference-Frame-Independent Quantum Key Distribution Using Fewer States},
  author = {Liu, Hongwei and Wang, Jipeng and Ma, Haiqiang and Sun, Shihai},
  journal = {Phys. Rev. Appl.},
  volume = {12},
  issue = {3},
  pages = {034039},
  numpages = {12},
  year = {2019},
  month = {Sep},
  publisher = {American Physical Society},
  doi = {10.1103/PhysRevApplied.12.034039},
  url = {https://link.aps.org/doi/10.1103/PhysRevApplied.12.034039}
}

@article{Imai_2021,
  title = {Bound Entanglement from Randomized Measurements},
  author = {Imai, Satoya and Wyderka, Nikolai and Ketterer, Andreas and G\"uhne, Otfried},
  journal = {Phys. Rev. Lett.},
  volume = {126},
  issue = {15},
  pages = {150501},
  numpages = {6},
  year = {2021},
  month = {Apr},
  publisher = {American Physical Society},
  doi = {10.1103/PhysRevLett.126.150501},
  url = {https://link.aps.org/doi/10.1103/PhysRevLett.126.150501}
}

@article{Kiukas_2017,
  title = {Continuous-variable steering and incompatibility via state-channel duality},
  author = {Kiukas, Jukka and Budroni, Costantino and Uola, Roope and Pellonp\"a\"a, Juha-Pekka},
  journal = {Phys. Rev. A},
  volume = {96},
  issue = {4},
  pages = {042331},
  numpages = {12},
  year = {2017},
  month = {Oct},
  publisher = {American Physical Society},
  doi = {10.1103/PhysRevA.96.042331},
  url = {https://link.aps.org/doi/10.1103/PhysRevA.96.042331}
}

@article{Werner_1989,
  title = {Quantum states with Einstein-Podolsky-Rosen correlations admitting a hidden-variable model},
  author = {Werner, Reinhard F.},
  journal = {Phys. Rev. A},
  volume = {40},
  issue = {8},
  pages = {4277--4281},
  numpages = {0},
  year = {1989},
  month = {Oct},
  publisher = {American Physical Society},
  doi = {10.1103/PhysRevA.40.4277},
  url = {https://link.aps.org/doi/10.1103/PhysRevA.40.4277}
}

@article{Sherman_1950,
   author = {Jack Sherman and Winifred J. Morrison},
   doi = {10.1214/aoms/1177729893},
   issn = {0003-4851},
   issue = {1},
   journal = {The Annals of Mathematical Statistics},
   month = {3},
   pages = {124-127},
   title = {Adjustment of an Inverse Matrix Corresponding to a Change in One Element of a Given Matrix},
   volume = {21},
   year = {1950}
}

@article{Tran_2015,
  title = {Quantum entanglement from random measurements},
  author = {Tran, Minh Cong and Daki\ifmmode \acute{c}\else \'{c}\fi{}, Borivoje and Arnault, Fran\ifmmode \mbox{\c{c}}\else \c{c}\fi{}ois and Laskowski, Wies\l{}aw and Paterek, Tomasz},
  journal = {Phys. Rev. A},
  volume = {92},
  issue = {5},
  pages = {050301(R)},
  numpages = {7},
  year = {2015},
  month = {Nov},
  publisher = {American Physical Society},
  doi = {10.1103/PhysRevA.92.050301},
  url = {https://link.aps.org/doi/10.1103/PhysRevA.92.050301}
}

@misc{Eisfeld_2026,
      title={Entanglement quantification with randomized measurements is maximally difficult}, 
      author={Julian Eisfeld and Nikolai Wyderka},
      year={2026},
      eprint={2604.15029},
      archivePrefix={arXiv},
      primaryClass={quant-ph},
      url={https://arxiv.org/abs/2604.15029}, 
}

@article{Bakhshinezhad_2024,
  title = {Trade-offs between precision and fluctuations in charging finite-dimensional quantum batteries},
  author = {Bakhshinezhad, Pharnam and Jablonski, Beniamin R. and Binder, Felix C. and Friis, Nicolai},
  journal = {Phys. Rev. E},
  volume = {109},
  issue = {1},
  pages = {014131},
  numpages = {17},
  year = {2024},
  month = {Jan},
  publisher = {American Physical Society},
  doi = {10.1103/PhysRevE.109.014131},
  url = {https://link.aps.org/doi/10.1103/PhysRevE.109.014131}
}

@article{Segal_2018,
  title = {Current fluctuations in quantum absorption refrigerators},
  author = {Segal, Dvira},
  journal = {Phys. Rev. E},
  volume = {97},
  issue = {5},
  pages = {052145},
  numpages = {9},
  year = {2018},
  month = {May},
  publisher = {American Physical Society},
  doi = {10.1103/PhysRevE.97.052145},
  url = {https://link.aps.org/doi/10.1103/PhysRevE.97.052145}
}

@article{Campisi2015,
   author = {Michele Campisi and Jukka Pekola and Rosario Fazio},
   doi = {10.1088/1367-2630/17/3/035012},
   issn = {1367-2630},
   issue = {3},
   journal = {New Journal of Physics},
   month = {3},
   pages = {035012},
   title = {Nonequilibrium fluctuations in quantum heat engines: theory, example, and possible solid state experiments},
   volume = {17},
   year = {2015}
}

@article{Das_2026,
   author = {Samir Das and Shishira Mahunta and Nikhil Gupt and Victor Mukherjee and Arnab Ghosh},
   doi = {10.1088/1402-4896/ae92ba},
   issn = {0031-8949},
   journal = {Physica Scripta},
   month = {7},
   title = {Fluctuations and optimal control in a floquet quantum thermal transistor},
   year = {2026}
}

@article{Horowitz_2020,
   author = {Jordan M. Horowitz and Todd R. Gingrich},
   doi = {10.1038/s41567-019-0702-6},
   issn = {1745-2473},
   issue = {1},
   journal = {Nature Physics},
   month = {1},
   pages = {15-20},
   title = {Thermodynamic uncertainty relations constrain non-equilibrium fluctuations},
   volume = {16},
   year = {2020}
}

@article{Barato_2015,
  title = {Thermodynamic Uncertainty Relation for Biomolecular Processes},
  author = {Barato, Andre C. and Seifert, Udo},
  journal = {Phys. Rev. Lett.},
  volume = {114},
  issue = {15},
  pages = {158101},
  numpages = {5},
  year = {2015},
  month = {Apr},
  publisher = {American Physical Society},
  doi = {10.1103/PhysRevLett.114.158101},
  url = {https://link.aps.org/doi/10.1103/PhysRevLett.114.158101}
}

@article{Jarzynski_2004,
  title = {Classical and Quantum Fluctuation Theorems for Heat Exchange},
  author = {Jarzynski, Christopher and W\'ojcik, Daniel K.},
  journal = {Phys. Rev. Lett.},
  volume = {92},
  issue = {23},
  pages = {230602},
  numpages = {4},
  year = {2004},
  month = {Jun},
  publisher = {American Physical Society},
  doi = {10.1103/PhysRevLett.92.230602},
  url = {https://link.aps.org/doi/10.1103/PhysRevLett.92.230602}
}

@article{Campisi_2011,
  title = {Colloquium: Quantum fluctuation relations: Foundations and applications},
  author = {Campisi, Michele and H\"anggi, Peter and Talkner, Peter},
  journal = {Rev. Mod. Phys.},
  volume = {83},
  issue = {3},
  pages = {771--791},
  numpages = {0},
  year = {2011},
  month = {Jul},
  publisher = {American Physical Society},
  doi = {10.1103/RevModPhys.83.771},
  url = {https://link.aps.org/doi/10.1103/RevModPhys.83.771}
}

@article{Campisi_2014,
   author = {Michele Campisi},
   doi = {10.1088/1751-8113/47/24/245001},
   issn = {1751-8113},
   issue = {24},
   journal = {Journal of Physics A: Mathematical and Theoretical},
   month = {6},
   pages = {245001},
   title = {Fluctuation relation for quantum heat engines and refrigerators},
   volume = {47},
   year = {2014}
}

@misc{Sarkar_2025,
      title={Fluctuation in energy extraction from quantum batteries: How open should the system be to control it?}, 
      author={Anindita Sarkar and Paranjoy Chaki and Priya Ghosh and Ujjwal Sen},
      year={2025},
      eprint={2505.16851},
      archivePrefix={arXiv},
      primaryClass={quant-ph},
      url={https://arxiv.org/abs/2505.16851}, 
}

@article{Imai_2023,
  title = {Work fluctuations and entanglement in quantum batteries},
  author = {Imai, Satoya and G\"uhne, Otfried and Nimmrichter, Stefan},
  journal = {Phys. Rev. A},
  volume = {107},
  issue = {2},
  pages = {022215},
  numpages = {12},
  year = {2023},
  month = {Feb},
  publisher = {American Physical Society},
  doi = {10.1103/PhysRevA.107.022215},
  url = {https://link.aps.org/doi/10.1103/PhysRevA.107.022215}
}

@inproceedings{leifer2007conditional,
  title     = {Conditional Density Operators and the Subjectivity of Quantum Operations},
  author    = {Leifer, M. S.},
  booktitle = {Foundations of Probability and Physics---4},
  editor    = {Adenier, Guillaume and Fuchs, Christopher A. and Khrennikov, Andrei Yu.},
  series    = {AIP Conference Proceedings},
  volume    = {889},
  pages     = {172--186},
  year      = {2007},
  doi       = {10.1063/1.2713456}
}

@article{leifer2013formulation,
  title   = {Towards a formulation of quantum theory as a causally neutral theory of {Bayesian} inference},
  author  = {Leifer, M. S. and Spekkens, Robert W.},
  journal = {Physical Review A},
  volume  = {88},
  number  = {5},
  pages   = {052130},
  year    = {2013},
  doi     = {10.1103/PhysRevA.88.052130}
}

@article{haapasalo2021quantum,
  title   = {Quantum marginal problem and incompatibility},
  author  = {Haapasalo, Erkka and Kraft, Tristan and Miklin, Nikolai and Uola, Roope},
  journal = {Quantum},
  volume  = {5},
  pages   = {476},
  year    = {2021},
  doi     = {10.22331/q-2021-06-15-476}
}

@article{verstraete2003normal,
  title   = {Normal forms and entanglement measures for multipartite quantum states},
  author  = {Verstraete, Frank and Dehaene, Jeroen and De Moor, Bart},
  journal = {Physical Review A},
  volume  = {68},
  number  = {1},
  pages   = {012103},
  year    = {2003},
  doi     = {10.1103/PhysRevA.68.012103}
}

@article{ruskai2003qubit,
  title   = {Qubit Entanglement Breaking Channels},
  author  = {Ruskai, Mary Beth},
  journal = {Reviews in Mathematical Physics},
  volume  = {15},
  number  = {6},
  pages   = {643--662},
  year    = {2003},
  doi     = {10.1142/S0129055X03001710}
}

@article{simonov2025activation,
  title   = {Activation of thermal states by coherently controlled thermalization processes},
  author  = {Simonov, Kyrylo and Roy, Saptarshi and Guha, Tamal and Zimbor{\'a}s, Zolt{\'a}n and Chiribella, Giulio},
  journal = {New Journal of Physics},
  volume  = {27},
  number  = {7},
  pages   = {074502},
  year    = {2025},
  doi     = {10.1088/1367-2630/ade5c4}
}

\widetext

\appendix

\section{proof of Theorem 1}
\label{app:prooft1}
        The most general two-qubit quantum-classical (qc) state $\xi^{SA}$ with Tr$_A (\xi^{SA}) = \xi^S$ is given by
       \begin{eqnarray}
        \xi^{SA} = p\rho \otimes \ketbra{\eta}{\eta} + (1-p) \sigma &\otimes& \ketbra{\eta^\perp}{\eta^\perp}, \nonumber \\
      \text{where } ~ p\rho + (1-p)\sigma &=& \xi^S.
      \label{eq:defxii}
    \end{eqnarray}
 Here $\ket \eta, \ket{\eta^\perp}$ is some orthonormal states.
    Let the auxiliary system be measured in the basis $\{\ket\psi,\ket{\psi^\perp}\}$ parameterized by $\theta$ and $\phi$ with $\{\ket 0,\ket 1\} \to \{ \ket \eta, \ket{\eta^\perp}\}$. Then the un-normalized post measurement states $\xi_{\psi (\psi^\perp)}^{SA} =$ $\Tr_A\big( (\mathbb{I}\otimes \ketbra{\psi (\psi^\perp)}{\psi (\psi^\perp)}) \xi^{SA} \big)$ when $\ket\psi (\ket{\psi^\perp})$ clicks are given by
    \begin{eqnarray}
\xi_\psi^{SA} &=& p\rho \cos^2 \frac{\theta}{2} + (1-p) \sigma \sin^2 \frac{\theta}{2}, \nonumber \\
&=& \sin^2 \frac{\theta}{2}\xi^S + \cos \theta ~p \rho \nonumber \\
\xi_{\psi^\perp}^{SA} &=& p\rho \sin^2 \frac{\theta}{2} + (1-p) \sigma \cos^2 \frac{\theta}{2}, \nonumber \\
&=&\sin^2 \frac{\theta}{2}\xi^S + \cos \theta ~(1-p)\sigma.
\label{eq:pms1}
    \end{eqnarray}
    where $\cos^2 \frac{\theta}{2} = |\langle\eta|\psi\rangle|^2$. 
    The net ergotropy that can be extracted from the post-measurement states is 
    \begin{eqnarray}
        E_D^{\xi^{SA}}(\theta) = {\tt Erg}(\xi_\psi^{SA}) + {\tt Erg}(\xi_{\xi_{\psi^\perp}}^{SA}).
    \end{eqnarray}
    Now by using Eq. \eqref{eq:pms1} and the convexity of ergotropy \cite{Bernards2019}, we have
\begin{eqnarray}
    E^{\xi^{SA}}_D(\theta) \leq 2 \sin^2 \frac{\theta}{2} {\tt Erg}(\xi^S) + \cos \theta ~[p {\tt Erg}(\rho) &+& (1-p) {\tt Erg}(\sigma)] \leq 2 \sin^2 \frac{\theta}{2} {\tt Erg}(\xi^S) + \cos \theta ~\text{Tr}(\xi^SH) . 
  \label{eq:ergtheta}
\end{eqnarray}
Now, the randomized daemonic ergotropy from Eq. \eqref{eq:rdequbits}  is given by

\begin{eqnarray}
    E^{\xi^{SA}}_D &=& \int_{0}^{\frac \pi2} d\theta ~\sin \theta E^{\xi^{SA}}_D(\theta) \nonumber \\
    &\leq& 2 {\tt Erg}(\xi^S) \Bigg( \int_{0}^{\frac{\pi}{2}} d\theta \sin \theta \sin^2 \frac{\theta}{2} \Bigg) + \text{Tr}(\xi^SH) \Bigg( \int_{0}^{\frac \pi2} d\theta ~\sin \theta \cos \theta \Bigg)  \nonumber \\
         &=& \frac12 \big[ \text{Tr} (\xi^SH) + {\tt Erg}(\xi^S) \big].
\end{eqnarray}
Now, $\Delta E^{\xi^{SA}}_D$ is simply obtained by subtracting ${\tt Erg}(\xi^S)$ from $E^{\xi^{SA}}_D$ which completes the proof.

\section{Proof of Theorem 2}
\label{app:proofoft2}
In the proof of Theorem~\ref{th:maxergqc}, we showed that the net ergotropy that can be extracted from the post-measurement states is 
    \begin{eqnarray}
        E_D^{\xi^{SA}}(\theta) = {\tt Erg}(\xi_\psi^{SA}) + {\tt Erg}(\xi_{\xi_{\psi^\perp}}^{SA}) = {\tt Erg}\Big(p\rho \cos^2 \frac{\theta}{2} + (1-p) \sigma \sin^2 \frac{\theta}{2}\Big)+{\tt Erg}\Big(p\rho \sin^2 \frac{\theta}{2} + (1-p) \sigma \cos^2 \frac{\theta}{2} \Big).
    \end{eqnarray}
    Identifying $u := \cos \theta$ and subtracting ${\tt Erg}(\xi^S)$, we have 
    \small
    \begin{eqnarray}
     \Delta E_D^{\xi^{SA}}(u) =    E_D^{\xi^{SA}}(u) - {\tt Erg}(\xi^S) =\frac12 \Big[{\tt Erg}\Big(p (1+u) \rho  + (1-p) (1-u) \sigma \Big)+{\tt Erg}\Big(p (1-u)\rho  + (1-p) (1+u) \sigma \Big)\Big] - {\tt Erg}(\xi^S).
    \end{eqnarray}
    \normalsize
   Now our quantity of interest is the function $\Delta E_D^{\xi^{SA}}(u)$ defined on the interval $u \in [0,1]$. This follows from  $\theta \in [0,\frac\pi2].$
   We now focus on some properties of $\Delta E_D^{\xi^{SA}}(u).$ 
   \begin{lemma}
       $\Delta E_D^{\xi^{SA}}(u)$ satisfies the following properties
       \begin{enumerate}
           \item $\Delta E_D^{\xi^{SA}}(u) \geq 0 ~\forall u \in [0,1].$

           \item $\Delta E_D^{\xi^{SA}}(u=0) = 0.$
           
           \item $\Delta E_D^{\xi^{SA}}(u)$ is a convex function of $u.$
       \end{enumerate}
       \label{lemma:gainprop}
   \end{lemma}
   \begin{proof}
       Since ${\tt Erg}(.)$ is a non-negative function, and we have $E_D^{\xi^{SA}}(u) \geq {\tt Erg}(\xi^S) ~\forall u \in [0,1],$ we have $\Delta E_D^{\xi^{SA}}(u) \geq 0 ~\forall u \in [0,1].$ Also, note that $E_D^{\xi^{SA}}(u=0) = 0.$ It is a direct consequence of the fact that for $u=0 ~(\theta = \pi/2),$ both the post-measurement state is identical to the marginal state of the system $\xi^S$.  Lastly, the convexity of Ergotropy on mixture of states directly translates to $\Delta E_D^{\xi^{SA}}(u)$ being a convex function of $u$. Subtraction of a constant term ${\tt Erg}(\xi^S)$, does not affect this property.
   \end{proof}

\begin{lemma}
Let \(g\) be a nonnegative function defined on \(u\in[0,1]\) satisfying the following properties: 
\begin{enumerate}
    \item \(g(u)\) is convex in \(u\),
    \item \(g(0)=0\).
\end{enumerate}
Then the following inequality holds: $\int_0^1 [g(u)]^2\,du
    \;\ge\;
    \frac{4}{3}
    \left[\int_0^1 g(u)\,du\right]^2.$
\end{lemma}

\begin{proof}
Let $ \int_0^1 g(u)\,du = m > 0.$
Now consider the function
\begin{equation}
    h(u)=2mu.
\end{equation}
Note that $ \int_0^1 h(u)\,du = m,$
so \(h(u)\) has the same mean as \(g(u)\).
Next, consider another convex function
\begin{equation}
    \ell(u)=g(u)-h(u),
\end{equation}
where by construction, we have $\int_0^1 \ell(u)\,du = 0.$
Then the 
properties $(1)$ and $(2)$ imply that there exists \(a\in(0,1)\) such that 
{\[
\ell(u)
\begin{cases}
\le 0, & u\in[0,a],\\[2mm]
>0, & u\in(a,1],
\end{cases}
\]
with $\int_0^a\ell(u)du=-\int_a^1\ell(u)du$}. Consequently, for any non-decreasing function \(k(u)\ge 0\), we get
\begin{equation}
    \int_0^1 du\, \ell(u)\,k(u)\ge 0.
\end{equation}
Since \(g(u)\ge 0\) is convex and \(g(0)=0\), \(g(u)\) is non-decreasing on \(u\in[0,1]\).
Moreover, as $h(u)$ by construction is a monotonically increasing function in $u$, $g(u)+h(u)$ is a non-decreasing function in $u$.
Let $k(u)=g(u)+h(u).$
Then we have
\begin{equation}
    \int_0^1 du\, \big(g(u)-h(u)  \big) \big(g(u)+h(u) \big) \ge 0.
\end{equation}
Therefore,
\begin{equation}
    \int_0^1  [g(u)]^2 ~du
    \ge
    \int_0^1  [h(u)]^2 ~du
    =
    4m^2 \int_0^1 u^2\,du
    =
    \frac{4m^2}{3}.
\end{equation}
Using \(m=\int_0^1 g(u)\,du\), we obtain
\begin{equation}
    \int_0^1 [g(u)]^2\,du
    \ge
    \frac{4}{3}
    \left[\int_0^1 g(u)\,du\right]^2.
\end{equation}
This completes the proof.
\end{proof}

\section{Proof of Theorem 4}
\label{app:proof t4}
Let
   $ \Pi_\pm(\mathbf n)
    =
    \frac{1}{2}
    \left(
        \mathbb I \pm \mathbf n\cdot\boldsymbol\sigma
    \right),
    ~
    \|\mathbf n\|=1,$
be a rank-one projective measurement on \(A\) performed on an arbitrary two-qubit state $\rho^{SA} = \frac{1}{4}
\left(
\mathbb I\otimes \mathbb I
+
\sum_{i=1}^{3} r_i \,\sigma_i \otimes \mathbb I
+
\sum_{j=1}^{3} a_j \, \mathbb I \otimes \sigma_j
+
\sum_{i,j=1}^{3} T_{ij}\,\sigma_i \otimes \sigma_j
\right),$ where $r_i = \operatorname{Tr}\!\left[\rho_{SA}\,(\sigma_i\otimes \mathbb I_A)\right], ~a_j = \operatorname{Tr}\!\left[\rho_{SA}\,(\mathbb I_S\otimes \sigma_j)\right],$ and $T_{ij}
=
\operatorname{Tr}\!\left[\rho_{SA}\,(\sigma_i\otimes\sigma_j)\right].$
The unnormalized post-measurement states on \(S\) corresponding to $\Pi_\pm(\mathbf n)$ is given by
\begin{equation}
    \rho^S_\pm(\mathbf n)
    =
    \frac{1}{4}
    \left[
        \bigl(1\pm \mathbf a\cdot \mathbf n\bigr)\mathbb I
        +
        \bigl(\mathbf r\pm T\mathbf n\bigr)\cdot\boldsymbol\sigma
    \right].
\end{equation}
For a qubit positive operator $ \omega
    =
    \frac{1}{2}
    \left(
        q\mathbb I+\mathbf s\cdot\boldsymbol\sigma
    \right),$
we have $ {\tt Erg}(\omega)
    =
    \frac{1}{2}
    \left(
        \|\mathbf s\|-s_z
    \right).$
Therefore,
\begin{eqnarray}
    E_D^{\rho^{SA}}(\mathbf n)
    =
    {\tt Erg}\bigl(\rho^S_+(\mathbf n)\bigr)
    +
    {\tt Erg}\bigl(\rho^S_-(\mathbf n)\bigr)
    =
    \frac{1}{4}
    \left(
        \|\mathbf r+T\mathbf n\|
        +
        \|\mathbf r-T\mathbf n\|
        -
        2r_z
    \right).
    \label{eq:EDn-classification}
\end{eqnarray}
In particular, the ancilla Bloch vector \(\mathbf a\) drops out.
Now, we define
\begin{equation}
    F(\mathbf n)
    :=
    \|\mathbf r+T\mathbf n\|
    +
    \|\mathbf r-T\mathbf n\|.
\end{equation}
Since \(E_D^{\rho^{SA}}(\mathbf n)\) is continuous in \(\mathbf n\), the condition
\(\mathcal F_D^{\rho^{SA}}=0\) is equivalent to \(E_D^{\rho^{SA}}(\mathbf n)\) being constant on the Bloch sphere, which in turn is equivalent to $F(\mathbf n)=2c$,
for some constant \(c\ge 0\) and for every \(\|\mathbf n\|=1\).

\medskip
\noindent\textit{Necessity.}
Assume \(\mathcal F_D^{\rho^{SA}}=0\) and \(\Delta E_D^{\rho^{SA}}>0\). Let
  $  F(\mathbf n)=2c$
for some constant \(c\), and by Eq.~\eqref{eq:EDn-classification}, 
   $ E_D^{\rho^{SA}}(\mathbf n)
    =
    \frac{1}{2}(c-r_z).$
On the other hand,
    ${\tt Erg}(\rho^S)
    =
    \frac{1}{2}
    \left(
        \|\mathbf r\|-r_z
    \right).$
Hence $\Delta E_D^{\rho^{SA}}
    =
    \frac{1}{2}
    \left(
        c-\|\mathbf r\|
    \right).$
Since the gain is positive, we must have $c > \|\mathbf r\|.$
Now we set $ \mathbf x:=T\mathbf n.$
As \(\mathbf n\) runs over the unit sphere, 
every such \(\mathbf x\) satisfies $ \|\mathbf x-\mathbf r\|+\|\mathbf x+\mathbf r\|=2c.$
Squaring it twice gives
\begin{equation}
    c^2\|\mathbf x\|^2
    -
    (\mathbf r\cdot \mathbf x)^2
    =
    c^2\bigl(c^2-\|\mathbf r\|^2\bigr).
    \label{eq:spheroid-quad}
\end{equation}
Substituting \(\mathbf x=T\mathbf n\), we obtain $ \mathbf n^{\mathsf T}
    T^{\mathsf T}
    \bigl(
        c^2 I_3-\mathbf r\mathbf r^{\mathsf T}
    \bigr)
    T
    \mathbf n
    =
    c^2\bigl(c^2-\|\mathbf r\|^2\bigr)$
for every unit vector \(\mathbf n\). Therefore, we arrive at the following identity
\begin{equation}
    T^{\mathsf T}
    \bigl(
        c^2 I_3-\mathbf r\mathbf r^{\mathsf T}
    \bigr)
    T
    =
    c^2\bigl(c^2-\|\mathbf r\|^2\bigr)I_3.
    \label{eq:key-matrix-id}
\end{equation}
By the positivity of gain condition $(c> \|\mathbf{r \|})$, the matrix $ c^2 I_3-\mathbf r\mathbf r^{\mathsf T}$
is positive definite, and the right-hand side of Eq.~\eqref{eq:key-matrix-id} is strictly positive. Hence \(T\) is invertible. We may therefore rewrite Eq.~\eqref{eq:key-matrix-id} as
\begin{equation}
    c^2 I_3-\mathbf r\mathbf r^{\mathsf T}
    =
    c^2\bigl(c^2-\|\mathbf r\|^2\bigr)
    (TT^{\mathsf T})^{-1}.
\end{equation}
Inverting both sides yields $    TT^{\mathsf T}
    =
    c^2\bigl(c^2-\|\mathbf r\|^2\bigr)
    \bigl(
        c^2 I_3-\mathbf r\mathbf r^{\mathsf T}
    \bigr)^{-1}.$
Now, by using the Sherman--Morrison formula \cite{Sherman_1950},
\begin{equation}
    \bigl(
        c^2 I_3-\mathbf r\mathbf r^{\mathsf T}
    \bigr)^{-1}
    =
    \frac{1}{c^2}I_3
    +
    \frac{1}{c^2\bigl(c^2-\|\mathbf r\|^2\bigr)}
    \mathbf r\mathbf r^{\mathsf T},
\end{equation}
we obtain $   TT^{\mathsf T}
    =
    \bigl(c^2-\|\mathbf r\|^2\bigr)I_3
    +
    \mathbf r\mathbf r^{\mathsf T}.$
By defining
   $ \beta^2:=c^2-\|\mathbf r\|^2>0,$
we arrive at $ TT^{\mathsf T}
    =
    \beta^2 I_3+\mathbf r\mathbf r^{\mathsf T}.$ This completes the necessary part of the proof.

\medskip
\noindent\textit{Sufficiency.}
Conversely, assume that $ TT^{\mathsf T}
    =
    \beta^2 I_3+\mathbf r\mathbf r^{\mathsf T}$
for some \(\beta>0\) where
    $c=\sqrt{\beta^2+\|\mathbf r\|^2}.$
Now, {the polar decomposition of the correlation matrix $T$ is given by $ T=(TT^{\mathsf T})^{1/2}O,$}
where \(O=(TT^{\mathsf T})^{-1/2}T\) is an orthogonal matrix. Since \(\mathbf n\mapsto O\mathbf n\) is merely a relabeling of the measurement direction on the sphere, the function \(F(\mathbf n)\) depends only on \((TT^{\mathsf T})^{1/2}\). 
Thus, for the purpose of evaluating \(F(\mathbf n)\), we may work in a frame in which \(\mathbf r=\|\mathbf r\|\hat z\). In that frame, $  (TT^{\mathsf T})^{1/2}
    =
    \mathrm{diag}(\beta,\beta,c)$.
So, up to the above reparametrization of measurement directions, we may take $ T
    =
    \mathrm{diag}(\beta,\beta,c).$
Then for \(\mathbf n=(n_x,n_y,n_z)\), we have
\begin{eqnarray}
    \|\mathbf r\pm T\mathbf n\|^2
    =
    \beta^2(n_x^2+n_y^2)
    +
    (\|\mathbf r\|\pm c n_z)^2
    =
    \beta^2(1-n_z^2)
    +
    \|\mathbf r\|^2
    +
    c^2n_z^2
    \pm
    2c\|\mathbf r\|n_z.
\end{eqnarray}
Using \(c^2-\beta^2=\|\mathbf r\|^2\), we get $  \|\mathbf r\pm T\mathbf n\|^2
    =
    (c\pm \|\mathbf r\|n_z)^2.$
Hence $  \|\mathbf r+T\mathbf n\|
    +
    \|\mathbf r-T\mathbf n\|
    =
    2c$
for every \(\mathbf n\). Therefore $ E_D^{\rho^{SA}}(\mathbf n)
    =
    \frac{1}{2}(c-r_z)$
is independent of the measurement direction, and consequently $\mathcal F_D^{\rho^{SA}}=0.$
Finally, we get $$ \Delta E_D^{\rho^{SA}}
    =
    \frac{1}{2}(c-r_z)
    -
    \frac{1}{2}(\|\mathbf r\|-r_z)
    =
    \frac{1}{2}(c-\|\mathbf r\|).$$
Since \(\beta>0\), we have \(c>\|\mathbf r\|\), and therefore
$\Delta E_D^{\rho^{SA}}>0.$
This completes the proof.

\section{Proof of Theorem 5}
\label{app:methods}
The proof is carried out in the following steps
\begin{enumerate}
   \item A lemma to prove $\Delta E_D$ is convex

    \item Evaluation of the extreme points of the separable states with fixed system marginal

    \item Compute and maximize RDE over the extreme points. 
\end{enumerate}

\subsection{Step 1: Convexity of randomized daemonic ergotropy and gain}

\begin{lemma}[Convexity of randomized daemonic ergotropy]
\label{lemma:convixityRDE}
Let $\rho_1^{SA}$ and $\rho_2^{SA}$ be two bipartite states, and let
\begin{align}
    \rho_\lambda^{SA}
    =
    \lambda \rho_1^{SA}
    +(1-\lambda)\rho_2^{SA},
    \qquad 0\leq\lambda\leq 1.
\end{align}
For a projective measurement
$\Pi=\{\Pi_i^A\}_i$ on the ancillary system $A$, define the
unnormalized conditional state of $S$ corresponding to outcome $i$ as
\begin{align}
    \rho^S(\Pi_i)
    :=
    \Tr_A\!\left[
    (\mathbb{I}_S\otimes\Pi_i^A)
    \rho^{SA}
    (\mathbb{I}_S\otimes\Pi_i^A)
    \right].
\end{align}
The daemonic ergotropy associated with the measurement $\Pi$ is
\begin{align}
    E_D^{\rho^{SA}}(\Pi)
    :=
    \sum_i
    \operatorname{Erg}\!\left[\rho^S(\Pi_i)\right],
\end{align}
where ergotropy is understood to be extended to subnormalized positive
operators by positive homogeneity. The randomized daemonic ergotropy,
obtained by averaging over measurements according to a normalized
measure $\mu(\Pi)$, is
\begin{align}
    E_D^{\rho^{SA}}
    :=
    \int d\Pi\,\mu(\Pi)\,
    E_D^{\rho^{SA}}(\Pi).
\end{align}
Then the randomized daemonic ergotropy is convex in $\rho^{SA}$, i.e.,
\begin{align}
    E_D^{\rho_\lambda^{SA}}
    \leq
    \lambda E_D^{\rho_1^{SA}}
    +(1-\lambda)E_D^{\rho_2^{SA}}.
\end{align}
\end{lemma}

\begin{proof}
We use the convexity of ergotropy,
\begin{align}
    \operatorname{Erg}
    \!\left[
    \lambda \sigma_1+(1-\lambda)\sigma_2
    \right]
    \leq
    \lambda \operatorname{Erg}[\sigma_1]
    +(1-\lambda)\operatorname{Erg}[\sigma_2].
\end{align}
For any fixed measurement $\Pi=\{\Pi_i^A\}_i$, linearity of the
post-measurement map gives
\begin{align}
    \rho_\lambda^S(\Pi_i)
    &=
    \Tr_A\!\left[
    (\mathbb{I}_S\otimes\Pi_i^A)
    \rho_\lambda^{SA}
    (\mathbb{I}_S\otimes\Pi_i^A)
    \right]
    \nonumber\\
    &=
    \lambda \rho_1^S(\Pi_i)
    +(1-\lambda)\rho_2^S(\Pi_i).
\end{align}
Hence, by convexity of ergotropy,
\begin{align}
    E_D^{\rho_\lambda^{SA}}(\Pi)
    &=
    \sum_i
    \operatorname{Erg}
    \!\left[
    \rho_\lambda^S(\Pi_i)
    \right]
    \nonumber\\
    &=
    \sum_i
    \operatorname{Erg}
    \!\left[
    \lambda\rho_1^S(\Pi_i)
    +(1-\lambda)\rho_2^S(\Pi_i)
    \right]
    \nonumber\\
    &\leq
    \sum_i
    \left[
    \lambda\operatorname{Erg}
    \!\left[\rho_1^S(\Pi_i)\right]
    +(1-\lambda)\operatorname{Erg}
    \!\left[\rho_2^S(\Pi_i)\right]
    \right]
    \nonumber\\
    &=
    \lambda E_D^{\rho_1^{SA}}(\Pi)
    +(1-\lambda)E_D^{\rho_2^{SA}}(\Pi).
\end{align}
Thus, the daemonic ergotropy is convex for every fixed measurement
$\Pi$. Averaging this inequality over the measurement ensemble gives
\begin{align}
    E_D^{\rho_\lambda^{SA}}
    &=
    \int d\Pi\,\mu(\Pi)\,
    E_D^{\rho_\lambda^{SA}}(\Pi)
    \nonumber\\
    &\leq
    \lambda
    \int d\Pi\,\mu(\Pi)\,
    E_D^{\rho_1^{SA}}(\Pi)
    +(1-\lambda)
    \int d\Pi\,\mu(\Pi)\,
    E_D^{\rho_2^{SA}}(\Pi)
    \nonumber\\
    &=
    \lambda E_D^{\rho_1^{SA}}
    +(1-\lambda)E_D^{\rho_2^{SA}}.
\end{align}
Now, by subtracting ${\tt Erg}(\rho^{S})$ from both sides, we get $\Delta E_D^{\rho_\lambda^{SA}}
\leq 
\lambda\Delta E_D^{\rho_1^{SA}}
+
(1-\lambda)\Delta E_D^{\rho_2^{SA}}.$
\end{proof}

\subsection{Step 2: Extreme points of separabble states with fixed marginals}

\begin{lemma}
\label{lemma:extreme}
    The extreme points of the set of separable states with fixed system marginal $\rho^S$ can be expressed as 
    \begin{eqnarray}
        p\,|\psi_1\rangle\langle\psi_1|
 \otimes|\phi_1\rangle\langle\phi_1|
+
(1-p)\,|\psi_2\rangle\langle\psi_2|
 \otimes|\phi_2\rangle\langle\phi_2|, ~~0 \leq p \leq 1,
    \end{eqnarray}
    such that $p|\psi_1\rangle\langle\psi_1|
+
(1-p)|\psi_2\rangle\langle\psi_2|
=
\rho^S$.
\end{lemma}
\begin{proof}
Assume first that \(\rho^S\) is mixed. Since \(S\) is a qubit,
\(\rho^S\) is then full rank. For every separable state
\(\rho^{SA}\) with system marginal $\rho^S$, consider the following bijective mapping
\begin{eqnarray}
    J_{\rho^{SA}}
:=
\left[(\rho^S)^{-1/2}\otimes I_A\right]
\rho^{SA}
\left[(\rho^S)^{-1/2}\otimes I_A\right].
\end{eqnarray}
The operator $J_{\rho^{SA}}$ satisfies
\[
\operatorname{Tr}_A J_{\rho^{SA}}
=
(\rho^S)^{-1/2}
\operatorname{Tr}_A(\rho^{SA})
(\rho^S)^{-1/2}
=
\mathbb I_S.
\]
The map \(\rho^{SA}\mapsto J_{\rho^{SA}}\) is an affine bijection, with
inverse
\begin{eqnarray}
    \rho^{SA}
=
\left[(\rho^S)^{1/2}\otimes  I_A\right]
J_{\rho^{SA}}
\left[(\rho^S)^{1/2}\otimes I_A\right].
\label{eq:inverse}
\end{eqnarray}
Such a map has been extensively considered in the literature in different contexts ranging from quantum marginal problems \cite{verstraete2003normal,Kiukas_2017, haapasalo2021quantum} to conditional states \cite{leifer2007conditional, leifer2013formulation}.

Since this is an invertible local deterministic transformation,
\(\rho^{SA}\) is separable if and only if \(J_{\rho^{SA}}\) is separable.
Thus \(J_{\rho^{SA}}\) can be identified as the Choi operator of an entanglement-breaking
qubit channel. Note that every entanglement-breaking qubit channel belongs to
the convex hull of classical--quantum (CQ) channels
\cite{ruskai2003qubit}. A CQ channel has
the form
\[
\Phi(X)
=
\langle e_1|X|e_1\rangle\,\tau_1^A
+
\langle e_2|X|e_2\rangle\,\tau_2^A,
\]
where \(\{|e_1\rangle,|e_2\rangle\}\) is an orthonormal basis of the
input qubit and $\tau_{1(2)}^A$ being density matrices.
By decomposing each \(\tau_i^A\) into pure states $\ket{\phi_i}$, every such channel
is itself a convex combination of CQ channels whose
two output states are pure. It is therefore sufficient to consider
Choi operators of the form
\[
J
= \Phi \otimes I(\ket{00}+\ket{11}) = 
|\overline{e_1}\rangle\langle\overline{e_1}|
 \otimes|\phi_1\rangle\langle\phi_1|
+
|\overline{e_2}\rangle\langle\overline{e_2}|
 \otimes|\phi_2\rangle\langle\phi_2|,
\]
where the bar denotes complex conjugation in the basis used to define
the Choi isomorphism.
Now, by using the inverse transformation in Eq. \eqref{eq:inverse}, we arrive at the form of the extreme points of every separable state with system marginal $\rho^S$
\[
\rho^{SA}
=
\sum_{i=1}^{2}
(\rho^S)^{1/2}
|\overline{e_i}\rangle\langle\overline{e_i}|
(\rho^S)^{1/2}
\otimes
|\phi_i\rangle\langle\phi_i|.
\]
Next, we define and make the following identifications
\[
p_i
:=
\langle\overline{e_i}|\rho^S|\overline{e_i}\rangle,
\qquad
|\psi_i\rangle
:=
\frac{(\rho^S)^{1/2}|\overline{e_i}\rangle}{\sqrt{p_i}}.
\]
Then \(p_1+p_2=1\), and $\rho^{SA}
=
p_1|\psi_1\rangle\langle\psi_1|
 \otimes|\phi_1\rangle\langle\phi_1|
+
p_2|\psi_2\rangle\langle\psi_2|
 \otimes|\phi_2\rangle\langle\phi_2|,$
 with $p_1|\psi_1\rangle\langle\psi_1|
+
p_2|\psi_2\rangle\langle\psi_2|
=
\rho^S.$

If \(\rho^S\) is pure, every extension of \(\rho^S\) is of the form
\(\rho^S\otimes\tau^A\). Since \(F\) is convex, a maximizer over
\(\tau^A\) can be chosen pure, which corresponds to the above form
with \(p=1\).
\end{proof}

\subsection{Step 3: Maximization of daemonic gain over the extreme points}

Since the daemonic gain is convex, as demonstrated in Lemma \ref{lemma:convixityRDE}, and the feasible set of separable states with fixed system marginal is convex, at least one maximizer of the daemonic gain is an extreme point of the feasible set. Therefore, from Lemma \ref{lemma:extreme}, it is sufficient to only consider extreme points of the set of separable convex states,
\begin{align}
\rho^{SA}
=
p |\psi_1\rangle\langle\psi_1| \otimes |\phi_1\rangle\langle\phi_1|
+
(1-p)|\psi_2\rangle\langle\psi_2| \otimes |\phi_2\rangle\langle\phi_2|,
\label{eq: D_extreme_state}
\end{align}
with a fixed system marginal
\begin{align}
p|\psi_1\rangle\langle\psi_1|
+
(1-p)|\psi_2\rangle\langle\psi_2|
=
\rho^S.
\end{align}
We consider $0<s<1$ in the following, where
\begin{align}
\rho^S
=
\frac{1}{2}
\left(
\mathbb I_S+\mathbf{r}\cdot\boldsymbol{\sigma}
\right),
\qquad
s:=\|\mathbf{r}\|.
\end{align}
The cases $s=0,1$ follow by taking the corresponding limits.

Let the states $\vert{}\psi_i\rangle$ and $\vert{}\phi_i\rangle$, for $i \in \{1, 2\}$, be characterized by the Bloch vectors $\mathbf{x}_i$ and $\mathbf{a}_i$, respectively. Since all these states are pure,
$\|\mathbf{x}_1\|
=
\|\mathbf{x}_2\|
=
\|\mathbf{a}_1\|
=
\|\mathbf{a}_2\|
=
1$. Therefore, the fixed marginal condition gives
\begin{align}
\mathbf{r}
=
p\mathbf{x}_1+(1-p)\mathbf{x}_2.
\end{align}
For convenience, we define
\begin{align}
\mathbf{d}:=
p\mathbf{x}_1-(1-p)\mathbf{x}_2.
\label{eq: D_eta_d}
\end{align}
Then $\mathbf{r}+\mathbf{d}
=
(1+\eta)\mathbf{x}_1$, and $\mathbf{r}-\mathbf{d}
=
(1-\eta)\mathbf{x}_2,$ where $\eta=2p-1$. Using $\|\mathbf{x}_1\|=\|\mathbf{x}_2\|=1$, adding and subtracting the squared norms of the above relations gives
\begin{align}
\mathbf{r}\cdot\mathbf{d}
&=
\eta,
&
\|\mathbf{d}\|^2+s^2
&=
1+\eta^2.
\label{eq: D_rd_relations}
\end{align}
Moreover, the fixed marginal condition implies
\begin{align}
s^2-\eta^2
=
2p(1-p)
\left(
1+\mathbf{x}_1\cdot\mathbf{x}_2
\right)
\geq 0,
\end{align}
and therefore $|\eta|\leq s$.

We next parametrize the two auxiliary Bloch vectors as
\begin{align}
\mathbf{b}
:=
\frac{\mathbf{a}_1+\mathbf{a}_2}{2},
\qquad
\mathbf{c}
:=
\frac{\mathbf{a}_1-\mathbf{a}_2}{2}.
\label{eq: D_bc_def}
\end{align}
Since $\|\mathbf{a}_1\|=\|\mathbf{a}_2\|=1$, these vectors satisfy
\begin{align}
\mathbf{b}\cdot\mathbf{c}
&=
0,
&
\|\mathbf{b}\|^2+\|\mathbf{c}\|^2
&=
1.
\label{eq: D_bc_relations}
\end{align}
The correlation matrix of the state in Eq.~\eqref{eq: D_extreme_state} is then
\begin{align}
T
&=
p\mathbf{x}_1\mathbf{a}_1^T
+
(1-p)\mathbf{x}_2\mathbf{a}_2^T
\nonumber\\
&=
\mathbf{r}\mathbf{b}^T
+
\mathbf{d}\mathbf{c}^T.
\label{eq: D_correlation_matrix}
\end{align}

Following Appendix~\ref{app:proof t4}, consider the projective measurement
\begin{align}
\Pi_{\pm}(\mathbf{n})
=
\frac{1}{2}
\left(
I\pm\mathbf{n}\cdot\boldsymbol{\sigma}
\right),
\qquad
\|\mathbf{n}\|=1,
\end{align}
on the auxiliary. For the Hamiltonian $H=|1\rangle\langle 1|$, the total daemonic ergotropy obtained from the two post-measurement outcomes is given by (see Eq.~\eqref{eq:EDn-classification}) 
\begin{align}
E_D^{\rho^{SA}}(\mathbf{n})
=
\frac{1}{4}
\left(
\|\mathbf{r}+T\mathbf{n}\|
+
\|\mathbf{r}-T\mathbf{n}\|
-
2r_z
\right).
\end{align}
Since
\begin{align}
\operatorname{Erg}(\rho^S)
=
\frac{1}{2}(s-r_z),
\end{align}
the corresponding daemonic gain is
\begin{align}
\Delta E_D^{\rho^{SA}}(\mathbf{n})
=
\frac{1}{4}
\left[
F(\mathbf{n})-2s
\right],
\qquad
F(\mathbf{n})
:=
\|\mathbf{r}+T\mathbf{n}\|
+
\|\mathbf{r}-T\mathbf{n}\|.
\label{eq: D_gain_F}
\end{align}
Importantly, the two terms in $F(\mathbf{n})$ correspond to the two measurement outcomes $\Pi_+$ and $\Pi_-$, respectively. Since the system marginal is fixed, maximizing the randomized daemonic gain is therefore equivalent to maximizing the Haar average value of $F(\mathbf{n})$.

Using Eq.~\eqref{eq: D_correlation_matrix}, let
$\beta:=\mathbf{b}\cdot\mathbf{n},
\gamma:=\mathbf{c}\cdot\mathbf{n}$. We then have
\begin{align}
T\mathbf{n}
=
\beta\mathbf{r}+\gamma\mathbf{d},
\end{align}
and consequently
\begin{align}
\|\mathbf{r}+T\mathbf{n}\|^2
&=
s^2(1+\beta)^2
+
2\eta\gamma(1+\beta)
+
(1+\eta^2-s^2)\gamma^2,
\label{eq: D_Rplus}
\\
\|\mathbf{r}-T\mathbf{n}\|^2
&=
s^2(1-\beta)^2
-
2\eta\gamma(1-\beta)
+
(1+\eta^2-s^2)\gamma^2.
\label{eq: D_Rminus}
\end{align}

\subsubsection{Optimality of the mixing probability $p=1/2$}
Let us first optimize over the mixing probability $p$. Since $\mathbf{b}\perp\mathbf{c}$, without loss of generality we can choose
\begin{align}
\mathbf{b}=B\hat{\mathbf{e}}_z,
\qquad
\mathbf{c}=C\hat{\mathbf{e}}_x,
\qquad
B^2+C^2=1,
\end{align}
and write a generic measurement direction as
\begin{align}
\mathbf{n}
=
x\hat{\mathbf{e}}_x
+
y\hat{\mathbf{e}}_y
+
z\hat{\mathbf{e}}_z,
\qquad
x^2+y^2+z^2=1.
\end{align}
Hence, we have
$\beta=Bz, \gamma=Cx$. The Haar measure is invariant under the reflection $x\mapsto-x$, which leaves $\beta$ unchanged and sends $\gamma\mapsto-\gamma$. We may therefore pair every measurement direction $(\beta,\gamma)$ with $(\beta,-\gamma)$.

To see the consequence of this pairing, let
\begin{align}
L_{\pm}:=1\pm\beta,
\qquad
A_{\pm}
:=
s^2L_{\pm}^2+(1-s^2)\gamma^2.
\end{align}
For either measurement outcome, Eqs.~\eqref{eq: D_Rplus} and \eqref{eq: D_Rminus} contain a pair of the form
\begin{align}
\sqrt{A_{\pm}+\eta^2\gamma^2+2L_{\pm}\eta\gamma}
+
\sqrt{A_{\pm}+\eta^2\gamma^2-2L_{\pm}\eta\gamma}.
\label{eq: D_pair}
\end{align}
{Setting $t=\eta\gamma$, we can write
\begin{align}
&
\left[
\sqrt{A_{\pm}+t^2+2L_{\pm}t}
+
\sqrt{A_{\pm}+t^2-2L_{\pm}t}
\right]^2
\nonumber\\
&=
2(A_{\pm}+t^2)
+
2\sqrt{(A_{\pm}+t^2)^2-4L_{\pm}^2t^2}.
\label{eq: D_pair_bound}
\end{align}
The second term in the above equation can be upper bounded by
\begin{align}
&\left[
(A_{\pm}+t^2)^2-4L_{\pm}^2t^2
\right]\leq (A_{\pm}-t^2)^2,
\end{align}
which follows from 
\begin{align}
&(A_{\pm}-t^2)^2
-
\left[
(A_{\pm}+t^2)^2-4L_{\pm}^2t^2
\right]
\nonumber\\
&=
4t^2(L_{\pm}^2-A_{\pm})\nonumber\\&=(1-s^2)(L_\pm^2-\gamma^2)\nonumber\\&=(1-s^2)[(B\pm z)^2+C^2y^2]\geq 0.
\end{align}
Finally, since, $(A_\pm-t^2)=s^2(L_{\pm}^2-\gamma^2)
+
(1-\eta^2)\gamma^2
\geq 0$, we have 
\begin{align}
\sqrt{A_{\pm}+\eta^2\gamma^2+2L_{\pm}\eta\gamma}
+
\sqrt{A_{\pm}+\eta^2\gamma^2-2L_{\pm}\eta\gamma}
\leq
2\sqrt{A_{\pm}}.
\label{eq: D_eta_bound}
\end{align}
By virtue of Eq.~\eqref{eq: D_eta_bound} and then averaging each pair $(\beta,\gamma)$ and $(\beta,-\gamma)$ over the Haar measure gives
\begin{align}
\overline{F}_{\eta}
\leq
\overline{F}_{\eta=0},
\end{align}
where $\overline{F}_\eta$ is the Haar averaged value of $F(\mathbf n)$ for a fixed value of $\eta$. Hence an optimal extreme point can always be chosen with
\begin{align}
\eta=0,
\qquad\text{or equivalently}\qquad
p=\frac{1}{2}.
\label{eq: D_p_half}
\end{align}

\subsubsection{Finding optimal Bloch vectors of auxiliary system}
For $\eta=0$, Eq.~\eqref{eq: D_rd_relations} reduces to
\begin{align}
\mathbf{r}\cdot\mathbf{d}
&=
0,
&
\|\mathbf{d}\|^2
&=
1-s^2,
\label{eq: D_d_opt}
\end{align}
which provide 
\begin{align}
F(\mathbf{n})
=
&
\sqrt{s^2(1+\beta)^2+(1-s^2)\gamma^2}+
\sqrt{s^2(1-\beta)^2+(1-s^2)\gamma^2}.
\label{eq: D_F_eta_zero}
\end{align}
Defining the concave function $f(q):=\sqrt{s^2+(1-s^2)q}$, we can rewrite Eq.~\eqref{eq: D_F_eta_zero} as 
\begin{align}
\frac{F(\mathbf{n})}{2}
&=
\frac{1+\beta}{2}
f\left(
\frac{\gamma^2}{(1+\beta)^2}
\right)
+
\frac{1-\beta}{2}
f\left(
\frac{\gamma^2}{(1-\beta)^2}
\right),\nonumber\\&\leq
f\left[
\frac{1+\beta}{2}
\frac{\gamma^2}{(1+\beta)^2}
+
\frac{1-\beta}{2}
\frac{\gamma^2}{(1-\beta)^2}
\right],
\nonumber\\
&=
\sqrt{
s^2
+
(1-s^2)
\frac{\gamma^2}{1-\beta^2}
},\nonumber\\&\leq\sqrt{s^2+(1-s^2)x^2}.
\label{eq: D_Jensen_bound}
\end{align}
The inequality in the last line follows from $\gamma^2/(1-\beta^2)\leq x^2$ since $1-B^2z^2-C^2\geq 0$. The saturation of the upper bound of $F(\mathbf n)$ requires $B=\|\mathbf b \|=0$ which from Eq.~\eqref{eq: D_bc_def} implies $\mathbf{a}_2=-\mathbf a_1$. We may therefore choose 
\begin{align}
|\phi_1\rangle=|0\rangle,
\qquad
|\phi_2\rangle=|1\rangle.
\end{align}
Combining this with $p=1/2$ the maximizing extreme point in the set of separable states can be written as
\begin{align}
\rho_*^{SA}
=
\frac{1}{2}
|\psi_+\rangle\langle\psi_+|
\otimes
|0\rangle\langle 0|
+
\frac{1}{2}
|\psi_-\rangle\langle\psi_-|
\otimes
|1\rangle\langle 1|,
\label{eq: D_opt_state}
\end{align}
where $\frac{1}{2}
|\psi_+\rangle\langle\psi_+|
+
\frac{1}{2}
|\psi_-\rangle\langle\psi_-|
=
\rho^S.$

\subsubsection{Maximum value of gain in RDE framework}
For Haar random qubit projective measurement, $x$ is uniformly distributed over $[-1,1]$. Following  Eq.~\eqref{eq: D_gain_F}, for the state in Eq.~\eqref{eq: D_opt_state} we calculate
\begin{align}
\Delta E_D^{\mathrm{sep,max}}(\rho^S)
&=
\frac{1}{2}
\left[
\int_0^1
dx\,
\sqrt{s^2+(1-s^2)x^2}
-s
\right]
\nonumber\\
&=
\frac{1}{2}
\left[
\frac{1}{2}
+
\frac{s^2}{2\sqrt{1-s^2}}
\operatorname{arcsinh}
\left(
\frac{\sqrt{1-s^2}}{s}
\right)
-s
\right].
\label{eq: D_sep_max_final}
\end{align}
This is precisely the value stated in Eq.~\eqref{eq:mrde}. The cases $s=0$ and $s=1$ are obtained by taking limits $0^+$ or $1^-$ of this expression. Therefore, the maximal randomized daemonic gain among all two-qubit separable states with fixed system marginal $\rho^S$ is attained by the quantum-classical state in Eq.~\eqref{eq: D_opt_state}. This completes the proof.
}

\section{Proof of Theorem 7}
\label{app:prooft6}
The proof of the gain part follows directly from the proof of Theorem 1.  For the polar cap of latitudinal extent $\theta_0$, from Eq. \eqref{eq:ergtheta}, the randomized daemonic ergotropy can be computed as
\begin{eqnarray}
    E^{\xi^{SA}}_D[\theta_0] &=& \frac{1}{(1-\cos \theta_0)}\int_{0}^{\theta_0} d\theta ~\sin \theta E^{\xi^{SA}}_D(\theta) \nonumber \\
    &\leq& \frac{1}{(1-\cos \theta_0)} \Bigg[ 2 {\tt Erg}(\xi^S) \Bigg( \int_{0}^{\theta_0} d\theta \sin \theta \sin^2 \frac{\theta}{2} \Bigg) + \text{Tr}(\xi^SH) \Bigg( \int_{0}^{\theta_0} d\theta ~\sin \theta \cos \theta \Bigg) \Bigg]  \nonumber \\
         &=&   \cos^2 \frac{\theta_0}{2} \text{Tr} (\xi^SH) + \sin^2 \frac{\theta_0}{2} {\tt Erg}(\xi^S).
\end{eqnarray}
Now, $\Delta E^{\xi^{SA}}_D[\theta_0]$ is simply obtained by subtracting ${\tt Erg}(\xi^S)$ from $E^{\xi^{SA}}_D$ which completes the proof.

The bound on the fluctuations again follows from the proof of Theorem 2 with the identification that the polar cap induces a smaller interval for consideration $u \in [\cos \theta_0,1].$ 
Motivated by the proof of Theorem 2, we define
\begin{eqnarray}
    \ell(u) = \frac{2mu}{1+\cos\theta_0}, ~\text{where} ~m =\frac{1}{1-\cos \theta_0}\int_{\cos \theta_0}^{1} ~du ~\ell(u) = \Delta E_D^{\xi^{SA}}[\theta_0].
\end{eqnarray}
It follows from Lemma \ref{lemma:gainprop} that \(g\) is convex and \(g(0)=0\), the ratio \(g(u)/u\) is nondecreasing for
\(u>0\). Therefore $ g(c)
    \le
    \frac{2cm}{1+c}
    =
    \ell(c).$ Now, define \(h(u):=g(u)-\ell(u)\). Then \(h\) is convex on \([c,1]\), has zero mean on
\([c,1]\), and satisfies \(h(c)\le 0\). Hence \(h\) changes sign at most once, from
negative to positive. Since both \(g\) and \(\ell\) are nondecreasing, so is
\(g+\ell\), and therefore
\begin{equation}
    \int_{\cos \theta_0}^1 h(u)\bigl(g(u)+\ell(u)\bigr)\,du \ge 0 ~\implies  \int_{\cos \theta_0}^1 g(u)^2\,du \ge \int_{\cos \theta_0}^1 \ell(u)^2\,du.
\end{equation}
Now, we have
\begin{eqnarray}
    \frac{1}{1-\cos \theta_0}\int_{\cos \theta_0}^1 \ell(u)^2\,du
    &=&
    \frac{1}{1-\cos \theta_0}\left(\frac{2m}{1+\cos \theta_0}\right)^2 \int_{\cos \theta_0}^1 u^2\,du =
    \frac{4m^2}{3}\cdot \frac{1+\cos \theta_0+\cos^2 \theta_0}{(1+\cos \theta_0)^2}.
\end{eqnarray}
Finally, combining everything and performing elementary trigonometric simplifications, we get
\begin{eqnarray}
    \bigl(\mathcal F_{D,\theta_0}^{\xi^{SA}}\bigr)^2
    &\ge& 
    m^2
    \left[
        \frac{4(1+\cos \theta_0+\cos^2 \theta_0)}{3(1+\cos \theta_0)^2}-1
    \right] =
    \frac{(1-\cos \theta_0)^2}{3(1+\cos \theta_0)^2}\,m^2.
\end{eqnarray}
Substituting \(m=\Delta E_{D}^{\xi^{SA}}[\theta_0]\) completes the proof.

\end{document}